\documentclass[conference,letterpaper]{IEEEtran}
\IEEEoverridecommandlockouts
\usepackage[utf8]{inputenc} 
\usepackage[T1]{fontenc}
\usepackage{url}
\usepackage{ifthen}
\usepackage{cite}
\usepackage[cmex10]{amsmath} 

\usepackage{amsfonts}
\usepackage{amssymb}
\usepackage{amsthm}
\usepackage{enumitem}
\usepackage{xcolor}
\usepackage{float}
\usepackage{bbm}
\usepackage{graphicx}
\usepackage{subcaption}

\usepackage{tikz}
\usepackage{tikz-3dplot}

\usepackage{pgfplots}
\pgfplotsset{compat=1.17}

\newtheorem{theorem}{Theorem}

\newtheorem{definition}{Definition}
\newtheorem{corollary}{Corollary}

\newtheorem{claim}{Claim}

\theoremstyle{definition}
\newtheorem{remark}{Remark}

\theoremstyle{definition} \newtheorem{example}{Example}

\def \extended {1}

\begin{document}
\title{On Han- and Shearer-Type Information Inequalities: Equality Conditions and Combinatorial Applications} 


\author{%
  \IEEEauthorblockN{Gunank Jakhar, Gowtham R. Kurri, Suryajith Chillara}
  \IEEEauthorblockA{International Institute of Information Technology, Hyderabad \\
                    Hyderabad, India.\\ 
                    Email: \{gunank.jakhar@research., gowtham.kurri@, suryajith.chillara@\}iiit.ac.in}
  \and
  \IEEEauthorblockN{Vinod M. Prabhakaran}
  \IEEEauthorblockA{Tata Institute of Fundamental Research\\
                    Mumbai, India.\\
                    Email: vinodmp@tifr.res.in}
}

\maketitle


\begin{abstract}
\if \extended 0%
  THIS PAPER IS ELIGIBLE FOR THE STUDENT PAPER AWARD.
\fi
  Madiman and Tetali (2010) established strong and weak fractional inequalities that subsume classical entropy subset inequalities such as Han’s inequality and Shearer’s lemma. In this work, we derive a decomposition of the gap in the strong fractional inequality and use it to establish the necessary and sufficient conditions for equality in the strong form of Han's inequality. As a consequence, we characterize the equality conditions for a strong Loomis-Whitney-type projection inequality of Balister and Bollob\'as (2012). Finally, we apply the weak fractional inequality to an extremal graph theory problem, recovering and extending results of Sason (2022) and Boucheron~\emph{et al.} (2013), whose approaches rely on Han's inequality.
\end{abstract}
\if \extended 0%
\textit{Complete proofs are provided in the extended version~\cite{JakharKCP26}}.
\fi

\section{Introduction}
Information inequalities have applications across diverse fields, including combinatorics~\cite{Pippenger99}, graph theory~\cite{Radhakrishnan03,Sason22}, communication complexity~\cite{BarYossefJKS02}, private information retrieval~\cite{sun2017capacity,yao2021capacity}. Classical examples include subadditivity of entropy, Han's inequality~\cite{Han78}, Shearer's lemma~\cite{ChungGFS86,Radhakrishnan03}. Han's inequality states that the joint entropy is upper bounded by a suitably weighted sum of the entropies of all $k$-sized subsets of $n$ random variables, for any $k\in[1:n]$. Shearer's lemma generalizes Han's inequality to arbitrary families of subsets.




Using Han's inequality, Radhakrishnan~\cite{Radhakrishnan03} gave an entropy proof of the Loomis--Whitney projection inequality~\cite{LoomisH49}, which bounds the cardinality of a finite set $S$ in $\mathbb{R}^d$ in terms of the cardinalities of its $(d-1)$-dimensional coordinate projections. It is well known that equality holds if and only if $S$ is a box, i.e., a Cartesian product of finite non-empty subsets of $\mathbb{R}$~\cite{EllisKFY16,CampiGS19}. This inequality has applications in databases~\cite{NgoPRR18} and matrix multiplication~\cite{Irony04}. Sason~\cite{Sason22} also applied Han's inequality to an extremal graph theory problem, generalizing an edge-isoperimetric inequality on the binary hypercube previously considered in \cite{BoucheronLM13}. 

Madiman and Tetali~\cite{MadimanT10} generalize Shearer's lemma by assigning arbitrary nonnegative weights to the members of the family, subject to fractional covering constraints. They also obtain a stronger upper bound by introducing suitable conditioning, generalizing an analogous strengthening of Shearer's lemma due to Kahn~\cite{Kahn01_partial_order}. We refer to these bounds as the weak and strong fractional inequalities, respectively.

Equality conditions and stability for the weak fractional inequality were studied in \cite{JakharKCP25}, with applications to multivariate mutual information and matrix inequalities. The strong fractional inequality has been used to derive upper bounds on the number of graph homomorphisms and the number of independent sets in an arbitrary graph~\cite[Section~V]{MadimanT10}. Kahn's earlier special case~\cite{Kahn01_partial_order} was used to show that a uniformly chosen cube-indexed random walk takes only a constant number of distinct values with exponentially high probability~\cite{Kahn01_partial_order}, resolving a conjecture of Benjamini~\emph{et~al.}~\cite{benjamini2000random}. 

In this work, we investigate the equality conditions and stability for the strong fractional inequality, and explore its applications to projection inequalities. We also apply the weak fractional inequality to an extremal graph theory problem. Our main contributions are as follows:
\begin{itemize}[leftmargin=*]
    \item We show that the gap in the strong form of Han's inequality admits a decomposition into a weighted sum of the conditional gaps in the weak form (Theorem~\ref{Gap_Decomposition}). Using this decomposition, we establish a stability result (Corollary~\ref{Strong_Han_Gap_Stability}) and characterize the necessary and sufficient conditions for equality in the strong form of Han's inequality (Corollary~\ref{Strong_Han_Gap_Equality}).
    \item We characterize equality conditions for a strong Loomis--Whitney-type projection inequality of Balister and Boll{o}b\'as~\cite[Theorem~2.2]{BalisterB12}. Specifically, equality holds if and only if every nonempty fiber over the first coordinate is a box, with corresponding sides having the same cardinalities across fibers (Theorem~\ref{StrongLW-ddim_Equality}). We also characterize the equality conditions for an analogous inequality for compact bodies (Theorem~\ref{StrongVolLW_Equality}).  
    \item Finally, we study an extremal graph theory problem motivated by confusability graphs of noisy channels, extending the formulation considered in \cite{Sason22}. Whereas Sason\cite{Sason22} uses Han's inequality for families consisting of all sets of a fixed size, we apply the weak fractional inequality to arbitrary set families. We obtain upper bounds on the number of edges (Theorem~\ref{ExtremalGraph}) that recover and extend the results of ~\cite[Theorem~4.2]{BoucheronLM13} and \cite[Theorem~2]{Sason22}. We further use these bounds to derive an edge-isoperimetric inequality concerning the influences of subsets of the binary hypercube (Theorem~\ref{InfluenceLowerBound}).
\end{itemize}

\section{Preliminaries}
\textit{Notation}: We use $[i : i+k]$ to represent the set $\{i, i+1, \ldots, i+k-1, i+k\}$, where $i,k\in\mathbb{N}$. For any set $A \subseteq [1:n]$, $< A$ denotes the set of all indices in $[1:n]$ that are less than every index in $A$. $A^{\text{c}}$ denotes the complement of a set $A$. $A \triangle B$ denotes the symmetric difference of sets $A$ and $B$, i.e., $A \triangle B = (A \setminus{B}) \cup (B \setminus{A})$. All the logarithmic quantities are with respect to base $2$.

\begin{definition}[Fractional Partition and Fractional Covering]
    \hspace{0.5em}
    \begin{enumerate}[leftmargin=*]
        \item Given a family $\mathcal{F}$ of subsets of $[1:n]$, a function $\gamma:~\mathcal{F}\rightarrow[0,\infty)$ is called a fractional partition if, for all $i \in [1:n]$, $\sum\limits_{S \in \mathcal{F}: i \in S} \gamma(S) = 1$.
        \item Given a family $\mathcal{F}$ of subsets of $[1:n]$, a function $\alpha:~\mathcal{F}\rightarrow[0,\infty)$ is called a fractional covering if, for all $i \in [1:n]$, $\sum\limits_{S \in \mathcal{F}: i \in S} \alpha(S) \geq 1$.        
    \end{enumerate}
\end{definition}

Let $X_1,\ldots,X_n$ be a collection of discrete random variables. Two fundamental properties of entropy, namely the chain rule and the fact that conditioning cannot increase entropy, underlie several classical inequalities relating the joint entropy $H(X_{[1:n]})$ to the subset entropies $H(X_S)$, $S\subseteq [1:n]$. Among these, Han's inequality~\cite{Han78} states that the joint entropy is bounded above by an appropriately normalized sum of the entropies of all $k$-element subsets of the random variables, i.e.,
\begin{align}\label{eq: Han}
    H(X_{[1:n]}) \leq \frac{1}{\binom{n-1}{k-1}} \sum_{S \subseteq [1:n]: |S| = k} H(X_S).
\end{align}
Shearer's lemma~\cite{ChungGFS86} extends this upper bound to arbitrary families of subsets of $[1:n]$. Specifically, let $\mathcal{F}$ be a family of subsets of $[1:n]$, and suppose that each index $i \in [1:n]$ belongs to at least $r$ members of $\mathcal{F}$. Then
\begin{align}\label{eq: Shearer}
    H(X_{[1:n]}) \leq \frac{1}{r} \sum_{S \in \mathcal{F}} H(X_S).
\end{align}
Madiman and Tetali~\cite{MadimanT10} generalize \eqref{eq: Shearer} by assigning arbitrary nonnegative weights to the members of $\mathcal{F}$, subject to fractional covering constraints. They also obtain a stronger upper bound by introducing suitable conditioning, as stated below. 
\begin{theorem}[\!\!{\cite{MadimanT10}}]\label{FractionalSL}
     Let $\mathcal{F}$ be a family of subsets of $[1:n]$, and let $\alpha:\mathcal{F}\rightarrow [0,\infty)$ be a fractional covering. For jointly distributed random variables $X_1, \ldots, X_n$, the following hold.
    \begin{enumerate}
        \item {[{Weak form}]} \begin{align}\label{eq: MT_UB}
            H(X_{[1:n]}) \leq \sum_{S \in \mathcal{F}} \alpha(S)H(X_S).
        \end{align}    
        \item {[{Strong form}]} \begin{align}\label{eq: MT_UB_Strong}
            H(X_{[1:n]}) \leq \sum_{S \in \mathcal{F}} \alpha(S)H(X_S \mid X_{<S}).
        \end{align}    
    \end{enumerate}
\end{theorem}

In \cite[Section V-A]{JakharKCP25}, the gap in inequality \eqref{eq: MT_UB} obtained by specializing the fractional covering $\alpha$ to a fractional partition $\gamma: \mathcal{F} \rightarrow (0,\infty)$ was introduced as a multivariate mutual information measure called the $(\mathcal{F},\gamma)$-mutual information, and denoted by $(\mathcal{F},\gamma)\text{-MI}(X_1; \ldots; X_n)$. As shown in \cite[Proposition~1]{JakharKCP25}, this measure generalizes total correlation~\cite{Watanabe60}, dual total correlation~\cite{Han75}, and shared information~\cite{CsiszarN04}. Analogously, we define its conditional version as follows.  
\begin{definition}
    Given a family $\mathcal{F}$ of subsets of $[1:n]$, a fractional partition $\gamma: \mathcal{F} \rightarrow [0,\infty)$, and jointly distributed random variables $X_1, X_2, \ldots, X_n, Z_1, Z_2, \ldots, Z_m$, the conditional $(\mathcal{F},\gamma)$-mutual information of $X_1, X_2, \ldots, X_n$ given $Z_1, Z_2, \ldots, Z_m$
    \begin{align}
        &(\mathcal{F},\gamma)\emph{-MI}(X_1; \ldots; X_n \mid Z_{[1:m]}) \nonumber \\
        &\quad \quad := \sum_{S \in \mathcal{F}} \gamma(S)H(X_S \mid Z_{[1:m]}) - H(X_{[1:n]} \mid Z_{[1:m]}).
    \end{align}
\end{definition}

\begin{theorem}\label{Conditional_FGamma_MI}
    Let $\mathcal{F}$ be a family of proper subsets of $[1:n]$ such that no two distinct indices $i,j\in[1:n]$ always appear together in the members of $\mathcal{F}$. Let $\gamma:\mathcal{F}\rightarrow (0,\infty)$ be a fractional partition. Then for jointly distributed random variables $X_1, X_2, \ldots, X_n, Z_1, Z_2, \ldots, Z_m$, we have $(\mathcal{F},\gamma)\emph{-MI}(X_1; \ldots; X_n \mid Z_{[1:m]}) \geq 0$, with equality if and only if $X_1, X_2, \ldots, X_n$ are conditionally independent given $Z_{[1:m]}$.
\end{theorem}
\begin{proof}
    The result follows by applying \eqref{eq: MT_UB} and the equality characterization in \cite[Proposition~1]{JakharKCP25} to the conditional distribution of $X_{[1:n]}$ given $Z_{[1:m]}=z_{[1:m]}$, for every $z_{[1:m]}$ with positive probability.
\end{proof}

When the fractional covering $\alpha$ in \eqref{eq: MT_UB_Strong} is specialized to a fractional partition $\gamma:\mathcal{F}\rightarrow(0,\infty)$, we denote the resulting gap by
\begin{align}\label{eq: Strong_Gap}
    \operatorname{Gap}(\mathcal{F},\gamma)
    &:=
    \sum_{S\in\mathcal{F}}
    \gamma(S)H(X_S\mid X_{<S})
    -H(X_{[1:n]}).
\end{align}
Note that, by \eqref{eq: MT_UB_Strong}, $\operatorname{Gap}(\mathcal{F},\gamma)\geq 0$.

Han's inequality and both the weak and strong forms in Theorem~\ref{FractionalSL} with $\alpha$ restricted to a fractional partition $\gamma$ remain valid for jointly distributed continuous random variables after replacing $H(\cdot)$ by $h(\cdot)$~\cite{MadimanT10}. The same is true for Theorem~\ref{Conditional_FGamma_MI}~\cite{JakharKCP25}.

\section{Stability Analysis and Equality Conditions for the Strong Form}
We first derive a decomposition of the gap defined in \eqref{eq: Strong_Gap} and use it to establish stability results and necessary and sufficient conditions for equality in the strong form. All the results in this section hold in both the discrete and continuous settings because the proofs use only the chain rule and the non-negativity of conditional mutual information. Let $\mathcal{C}_k$ denote the complete $k$-uniform family on $[1:n]$, i.e., $\mathcal{C}_k=\{S\subseteq [1:n]: |S|=k\}$.
\begin{theorem}\label{Gap_Decomposition}
   Let $\mathcal{F}$ be a family of subsets of $[1:n]$, and let $\gamma:\mathcal{F}\rightarrow [0,\infty)$ be a fractional partition. Then, for jointly distributed random variables $X_1, \ldots, X_n$, the following hold.
    \begin{enumerate}[leftmargin=*]
        \item $\emph{Gap}(\mathcal{F},\gamma)$ decomposes as follows.
        \begin{align}
            &\hspace{-1.3em}\emph{Gap}(\mathcal{F},\gamma) \nonumber \\ 
            &\hspace{-1.5em}= \sum_{S\in\mathcal{F}}\gamma(S) \sum_{i\in S} I\left( X_i;X_{[1:i-1] \cap S^{\mathrm{c}} \setminus{(<S)}} | X_{([1:i-1] \cap S) \cup (<S)} \right).
        \end{align}
        \item If $\mathcal{F}=\mathcal{C}_k$, $k\in[2:n-1]$\footnote{The gap $\operatorname{Gap}(\mathcal{C}_k,\gamma)$ is equal to zero for $k=1$ by the chain rule and for $k=n$ directly from its definition.} and $\gamma(S)=\frac{1}{\binom{n-1}{k-1}}$ for all $S \in \mathcal{C}_k$, then
        \begin{align}
            &\emph{Gap}(\mathcal{C}_k,\gamma) \nonumber \\
            &= \sum_{j = 1}^{n-k} \frac{\binom{n-j-1}{k-2}}{\binom{n-1}{k-1}}(\mathcal{F}_j,\gamma_j)\emph{-MI}(X_{j+1}; \ldots; X_n \mid X_{[1:j]}),
        \end{align}
        where, for each $j \in [1:n-k]$, $\mathcal{F}_j:=\{A\subseteq[j+1:n]:|A|=k-1\}$ and $\gamma_j(A):=\frac{1}{\binom{n-j-1}{k-2}}$ for all $A\in\mathcal{F}_j$.
    \end{enumerate}
\end{theorem}

Part 1 of Theorem~\ref{Gap_Decomposition} follows from the chain rule of entropy and the defining property of a fractional partition. Part 2 exploits the combinatorial symmetry of the complete $k$-uniform family $\mathcal{C}_k$ to express $\operatorname{Gap}(\mathcal{C}_k,\gamma)$ as a weighted sum of the conditional weak-form gaps, namely conditional $(F,\gamma)$-mutual information quantities. 
\if \extended 1%
A detailed proof of Theorem~\ref{Gap_Decomposition} is provided in Appendix~\ref{appendix:a}.
\fi

\begin{corollary}\label{Strong_Han_Gap_Stability}
    Let $k\in[2:n-1]$, and let $\gamma:\mathcal{C}_k\rightarrow (0,\infty)$ be the fractional partition defined by $\gamma(S) = \frac{1}{\binom{n-1}{k-1}}$ for all $S \in \mathcal{C}_k$. For jointly distributed random variables $X_1, \ldots, X_n$ and any $\varepsilon \geq 0$, if $\operatorname{Gap}(\mathcal{C}_k,\gamma) \leq \varepsilon$, then
    \begin{align}
        (\mathcal{F},\gamma')\emph{-MI}(X_2; \ldots; X_n \mid X_1) \leq \frac{n-1}{k-1}\varepsilon,
    \end{align}
    where $\mathcal{F}=\{A\subseteq [2:n]: |A|=k-1\}$ and $\gamma'(A) = \frac{1}{\binom{n-2}{k-2}},$ for all $A \in \mathcal{F}$.
\end{corollary}

Corollary~\ref{Strong_Han_Gap_Stability} shows that near equality in the strong form of Han's inequality forces the conditional dependence among $X_2,\dots,X_n$ given $X_1$, as measured by their conditional $(\mathcal{F},\gamma')$-mutual information, to be small.
\if \extended 1%
A detailed proof of Corollary~\ref{Strong_Han_Gap_Stability} is provided in Appendix~\ref{appendix:b}.
\fi

\begin{corollary}\label{Strong_Han_Gap_Equality}
    Let $k\in[2:n-1]$, and let $\gamma:\mathcal{C}_k\rightarrow (0,\infty)$ be the fractional partition defined by $\gamma(S) = \frac{1}{\binom{n-1}{k-1}}$ for all $S \in \mathcal{C}_k$. For jointly distributed random variables $X_1 \ldots, X_n$, $\operatorname{Gap}(\mathcal{C}_k,\gamma) = 0$ if and only if $X_2, \ldots, X_n$ are conditionally independent given $X_1$.
\end{corollary}

For every $k\in[2:n-1]$, equality in Han's inequality \eqref{eq: Han} holds if and only if $X_1, \ldots, X_n$ are independent~\cite{Han78}. Thus, the equality condition does not depend on $k$. Corollary~\ref{Strong_Han_Gap_Equality} establishes the corresponding conclusion for the strong form of Han's inequality: its equality condition also does not depend on $k$, namely that $X_2,\dots,X_n$ are conditionally independent given $X_1$. 
\if \extended 1%
A detailed proof of Corollary~\ref{Strong_Han_Gap_Equality} is provided in Appendix~\ref{appendix:c}.
\fi

\section{Applications to Strong Loomis--Whitney-Type Projection Inequalities}
\subsection{Finite Sets}
The Loomis-Whitney inequality~\cite{LoomisH49} bounds the cardinality of a finite set in $\mathbb{R}^d$ in terms of the cardinalities of its coordinate projections. Let $d\geq 2$ and $S\subseteq \mathbb{R}^d$ be a set of $n$ distinct points. For each $i\in[1:d]$, let $n_{-i}$ denote the cardinality of the coordinate projection obtained by deleting the $i^{\text{th}}$-coordinate. Then
\begin{align}\label{eqn:LW}
    n^{d-1} \leq \prod_{i = 1}^d n_{-i},
\end{align}
with equality if and only if $S$ is a box, i.e., a Cartesian product of finite non-empty subsets of $\mathbb{R}$~\cite{EllisKFY16,CampiGS19}. 

 The following theorem is a strengthening of the Loomis--Whitney inequality, established in \cite{BalisterB12}.

\begin{theorem}[\!{\cite[A Finite-Set Form of Theorem~2.2]{BalisterB12}}]\label{StrongLW-ddim}
    Let $d\geq 3$, and let $S\subseteq \mathbb{R}^d$ be a set of $n$ distinct points. For each $i\in[2:d]$, let $n_{-i}$ denote the cardinality of the coordinate projection of $S$ obtained by deleting the $i^{\text{th}}$-coordinate. For $x_1\in\mathbb{R}$, define the slice $S_{x_1}=\{(x_2,\dots,x_d):(x_1,x_2,\dots,x_d)\in S\}$, and let 
    \begin{align}
        n^{\emph{max}}_{-1}=\max_{x_1\in\mathbb{R}}|S_{x_1}|
    \end{align}
    denote the largest slice cardinality. Then
    \begin{align}\label{eq:StrongLW-ddim}
            n^{d-1}\leq \left(\prod_{i=2}^{d}n_{-i}\right)n^{\emph{max}}_{-1}.
    \end{align}
\end{theorem}

The inequality \eqref{eq:StrongLW-ddim} is tighter than \eqref{eqn:LW} because $n^{\text{max}}_{-1}\leq n_{-1}$.  Balister and Bollob\'as~\cite{BalisterB12} establish \eqref{eq:StrongLW-ddim} using arguments based on H\"older's inequality. For our equality analysis, motivated by \cite{Radhakrishnan03}, we use an information-theoretic proof based on the strong form in \eqref{eq: MT_UB_Strong}. This approach yields the following characterization for equality in \eqref{eq:StrongLW-ddim}.

\begin{theorem}\label{StrongLW-ddim_Equality}
    Let $S$, $n$, $n_{-i}$ for $i\in[2:d]$, $S_{x_1}$ for $x_1\in\mathbb{R}$, and $n^{\emph{max}}_{-1}$ be as in Theorem~\ref{StrongLW-ddim}. Then 
    \begin{align}\label{eq:StrongLW_equality}
            n^{d-1} = \left(\prod_{i=2}^{d}n_{-i}\right)n^{\emph{max}}_{-1}
    \end{align}
    if and only if, for every $x_1$ with $S_{x_1}\neq \emptyset$, there exist finite sets $\mathcal{X}^{(2)}_{x_1},\ldots,\mathcal{X}^{(d)}_{x_1}\subseteq \mathbb{R}$ satisfying
    \begin{align}
        S_{x_1} = \prod_{i=2}^{d} \mathcal{X}^{(i)}_{x_1},
    \end{align}
    and $|\mathcal{X}^{(i)}_{x_1}| = |\mathcal{X}^{(i)}_{x_1'}|$, $i\in[2:d]$, for all $x_1,x_1'$ such that $S_{x_1},S_{{x_1'}}\neq \emptyset$. 
\end{theorem}

\if \extended 1%
A detailed proof of Theorem~\ref{StrongLW-ddim_Equality} is provided in Appendix~\ref{appendix:d}.
\fi
Unlike equality in \eqref{eqn:LW}, equality in \eqref{eq:StrongLW_equality} allows the sets $\mathcal{X}^{(i)}_{x_1}$ to vary with $x_1$, provided that their corresponding cardinalities remain constant across slices. Thus, every nonempty slice $S_{x_1}$ is a box, although $S$ itself need not be a box. Fig.~\ref{fig: StrongLW_equality} illustrates the equality conditions. 

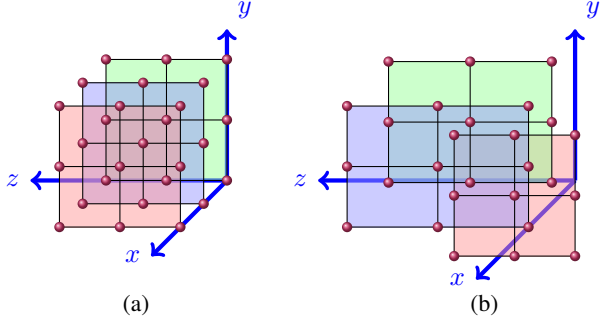
\begin{figure}[t]
    \centering
    \makebox[0.48\columnwidth][c]{
        \begin{minipage}{\linewidth}
        \begin{tikzpicture}[scale=0.8,
            cube/.style={very thick,black},
            grid/.style={very thin,gray},
            axis/.style={->,blue,ultra thick},
            rotated axis/.style={->,blue,ultra thick}
        ]
            \path[use as bounding box] (-3.8,-1.5) rectangle (0.8,3.0);
            \draw[rotated axis] (0,0,0) -- (-3.25,0,0)
                node[anchor=east] {$z$};
            \draw[rotated axis] (0,0,0) -- (0,2.5,0)
                node[anchor=south west] {$y$};
            \draw[rotated axis] (0,0,0) -- (0,0,3.25)
                node[anchor=east] {$x$};
        
            \foreach \z in {0,1,2} {
        
                \def\planecolor{green!40}
                \ifnum\z=1
                    \def\planecolor{blue!40}
                \fi
                \ifnum\z=2
                    \def\planecolor{red!40}
                \fi
        
                \fill[\planecolor,opacity=0.5]
                    (0,0,\z) --
                    (-2,0,\z) --
                    (-2,2,\z) --
                    (0,2,\z) -- cycle;
            }
        
            \foreach \x in {0,-1,-2}
                \foreach \y in {0,1,2}
                    \foreach \z in {0,1,2} {
        
                        \ifnum\x>-2
                            \draw[black]
                                (\x,\y,\z) -- (\the\numexpr\x-1\relax,\y,\z);
                        \fi
        
                        \ifnum\y<2
                            \draw[black]
                                (\x,\y,\z) -- (\x,\the\numexpr\y+1\relax,\z);
                        \fi
        
                        \shade[ball color=purple!80]
                            (\x,\y,\z) circle (0.085cm);
                    }
        
        \end{tikzpicture}
        \centering
        \par\smallskip
        {\small (a)}
        \end{minipage}
    }
    \hfill
    \makebox[0.48\columnwidth][c]{
        \centering
        \begin{minipage}{\linewidth}
        \begin{tikzpicture}[
            scale=0.8,
            cube/.style={very thick,black},
            grid/.style={very thin,gray},
            axis/.style={->,blue,ultra thick},
            rotated axis/.style={->,blue,ultra thick}
        ]
        
            \path[use as bounding box] (-3.8,-1.5) rectangle (0.8,3);
        
        
            \draw[rotated axis]
                (0,0,0) -- (-4.25,0,0) node[anchor=east] {$z$};
        
            \draw[rotated axis]
                (0,0,0) -- (0,2.5,0) node[anchor=south west] {$y$};
        
            \draw[rotated axis]
                (0,0,0) -- (0,0,4.25) node[anchor=east] {$x$};

        
            \fill[red!40,opacity=0.5]
                (0,-1.25,0) --
                (-2,-1.25,0) --
                (-2,0.75,0) --
                (0,0.75,0) -- cycle;
        
            \fill[green!40,opacity=0.5]
                (0,0.35,1) --
                (-2.7,0.35,1) --
                (-2.7,2.35,1) --
                (0,2.35,1) -- cycle;
        
            \fill[blue!40,opacity=0.5]
                (0,0,2) --
                (-3,0,2) --
                (-3,2,2) --
                (0,2,2) -- cycle;

        
            \foreach \x in {0,-1,-2} {
                \foreach \y in {0,1,2} {
        
                    \pgfmathsetmacro{\yy}{\y-1.25}
        
                    \ifthenelse{\lengthtest{\x pt > -2 pt}}{
                        \pgfmathsetmacro{\xn}{\x-1}
                        \draw[black] (\x,\yy,0) -- (\xn,\yy,0);
                    }{}
        
                    \ifthenelse{\y<2}{
                        \pgfmathsetmacro{\yn}{\yy+1}
                        \draw[black] (\x,\yy,0) -- (\x,\yn,0);
                    }{}
        
                    \shade[ball color=purple!80]
                        (\x,\yy,0) circle (0.085cm);
                }
            }

        
        
        
        
        
            
            
            \foreach \x in {0,-1.35,-2.7} {
                \foreach \y in {0,1,2} {
            
                    \pgfmathsetmacro{\yy}{\y+0.35}
            
                    \ifthenelse{\lengthtest{\x pt > -2.7 pt}}{
                        \pgfmathsetmacro{\xn}{\x-1.35}
                        \draw[black] (\x,\yy,1) -- (\xn,\yy,1);
                    }{}
            
                    \ifthenelse{\lengthtest{\y pt < 2 pt}}{
                        \pgfmathsetmacro{\yn}{\yy+1}
                        \draw[black] (\x,\yy,1) -- (\x,\yn,1);
                    }{}
            
                    \shade[ball color=purple!80]
                        (\x,\yy,1) circle (0.085cm);
                }
            }

        
            \foreach \x in {0,-1.5,-3} {
                \foreach \y in {0,1,2} {
        
                    \ifthenelse{\lengthtest{\x pt > -3 pt}}{
                        \pgfmathsetmacro{\xn}{\x-1.5}
                        \draw[black] (\x,\y,2) -- (\xn,\y,2);
                    }{}
        
                    \ifthenelse{\y<2}{
                        \pgfmathsetmacro{\yn}{\y+1}
                        \draw[black] (\x,\y,2) -- (\x,\yn,2);
                    }{}
        
                    \shade[ball color=purple!80]
                        (\x,\y,2) circle (0.085cm);
                }
            }
        
        \end{tikzpicture}
        \centering
        \par\smallskip
        {\small (b)}
        \end{minipage}
    }
    \caption{(a) and (b) illustrate equality in \eqref{eqn:LW} and \eqref{eq:StrongLW_equality}, respectively.}
    \label{fig: StrongLW_equality}
\end{figure}

\subsection{Compact Bodies}
 As in~\cite{BalisterB12}, a \emph{body} in $\mathbb{R}^d$ is a compact set that is the closure of its interior. For $j\in[1:d]$, we write $V_{j}(\cdot)$ for the $j$-dimensional Lebesgue measure, also referred to as the $j$-dimensional volume. The following theorem gives an analogous projection inequality for bodies.

\begin{theorem}[\!\!{\cite[Theorem~2.2]{BalisterB12}}]\label{StrongVolLW}
Let $d\geq 3$, and let $K\subseteq\mathbb{R}^d$ be a body. For each $i\in[2:d]$, let $K_{-i}$ denote the coordinate projection of $K$ obtained by deleting the $i^{\text{th}}$-coordinate. For $x_1\in\mathbb{R}$, define $K_{x_1}=\{(x_2,\dots,x_d):(x_1,x_2,\dots,x_d)\in K\}$, and let 
\begin{align}\label{eqn:esssup}
    M_{-1}(K)=\operatorname*{ess\,sup}_{x_1\in\mathbb{R}}V_{d-1}(K_{x_1}),
\end{align}
where the essential supremum\footnote{Although \cite[Theorem~2.2]{BalisterB12} is stated using a supremum, its proof yields the stronger bound with the essential supremum.} is taken with respect to Lebesgue measure. Then
\begin{align}\label{eq:StrongVolLW}
        V_d(K)^{d-1}\leq \left(\prod_{i=2}^{d}V_{d-1}(K_{-i})\right)M_{-1}(K).
\end{align}
\end{theorem}

The following theorem establishes the conditions for equality to hold in \eqref{eq:StrongVolLW}.

\begin{theorem}\label{StrongVolLW_Equality}
    Let $K, K_{-i},$ for $i \in [2:d], K_{x_1}$, and $M_{-1}(K)$ be as in Theorem~\ref{StrongVolLW}. Then
    \begin{align}\label{eq:StrongVolLW_Equality}
        V_d(K)^{d-1} = \left(\prod_{i=2}^{d}V_{d-1}(K_{-i})\right)M_{-1}(K).
    \end{align}
    if and only if there exist constants $\ell_2,\dots,\ell_d>0$ such that, for Lebesgue-almost every $x_1\in\mathbb{R}$ with $V_{d-1}(K_{x_1})>0$, there exist measurable sets  $\mathcal{A}^{(2)}_{x_1},\ldots,\mathcal{A}^{(d)}_{x_1}\subseteq \mathbb{R}$ satisfying
    \begin{align}
        V_{d-1}\left( K_{x_1} \triangle \prod_{i = 2}^d \mathcal{A}^{(i)}_{x_1} \right) = 0,
    \end{align}
and for each $i\in[2:d]$, $V_1(\mathcal{A}_{x_1}^{(i)})=\ell_i$.
\end{theorem}
\if \extended 1%
A detailed proof of Theorem~\ref{StrongVolLW_Equality} is provided in Appendix~\ref{appendix:e}.
\fi
Fig.~\ref{fig: VolumeStrongLW_equality} illustrates examples of bodies satisfying these equality conditions.

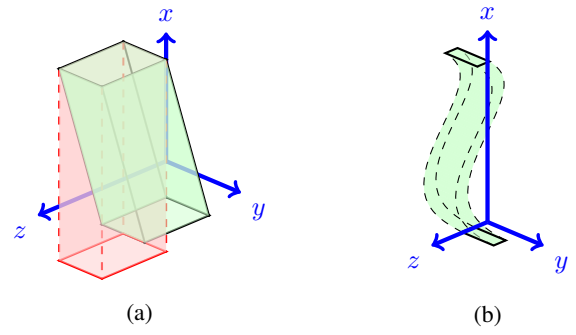
\begin{figure}[t]
    \centering
    \makebox[0.48\columnwidth][c]{
        \begin{minipage}{\linewidth}
        \tdplotsetmaincoords{65}{135}
        \begin{tikzpicture}[tdplot_main_coords, scale=0.4,
            cube/.style={very thick,black},
            grid/.style={very thin,gray},
            axis/.style={->,blue,ultra thick},
            rotated axis/.style={->,blue,ultra thick}]
        
        \draw[rotated axis] (0,0,0) -- (6,0,0) node[anchor=north east]{$z$};
        \draw[rotated axis] (0,0,0) -- (0,3.5,0) node[anchor=north west]{$y$};
        \draw[rotated axis] (0,0,0) -- (0,0,4.7) node[anchor=south]{$x$};
        
        \coordinate (P) at (2,2,0);
        \coordinate (Q) at (5,2,0);
        \coordinate (R) at (5,4,0);
        \coordinate (S) at (2,4,0);
        
        \coordinate (T) at (5,3,7);
        \coordinate (U) at (8,3,7);
        \coordinate (V) at (8,5,7);
        \coordinate (W) at (5,5,7);
        
        \coordinate (Txy) at (5,3,0);
        \coordinate (Uxy) at (8,3,0);
        \coordinate (Vxy) at (8,5,0);
        \coordinate (Wxy) at (5,5,0);
        
        \draw[thick] (P)--(Q)--(R)--(S)--cycle;
        
        \draw[thick] (T)--(U)--(V)--(W)--cycle;
        
        \draw[thick] (P)--(T);
        \draw[thick] (Q)--(U);
        \draw[thick] (R)--(V);
        \draw[thick] (S)--(W);
        
        \draw[dashed, red, thick] (T) -- (Txy);
        \draw[dashed, red, thick] (U) -- (Uxy);
        \draw[dashed, red, thick] (V) -- (Vxy);
        \draw[dashed, red, thick] (W) -- (Wxy);
        
        \draw[red, thick] (Txy)--(Uxy)--(Vxy)--(Wxy)--cycle;
        
        \foreach \p/\lab in {
        P/{P(2,2,0)},
        Q/{Q(5,2,0)},
        R/{R(5,4,0)},
        S/{S(2,4,0)},
        T/{T(5,3,4)},
        U/{U(8,3,4)},
        V/{V(8,5,4)},
        W/{W(5,5,4)}
        }
        {
            \fill (\p) circle (1.5pt);
        }
        
        \foreach \p in {Txy,Uxy,Vxy,Wxy}
        {
            \fill[red] (\p) circle (1.2pt);
        }
        
        \fill[red!20, opacity=0.5]
        (T) -- (U) -- (Uxy) -- (Txy) -- cycle;
        
        \fill[red!20, opacity=0.5]
        (U) -- (V) -- (Vxy) -- (Uxy) -- cycle;
        
        \fill[red!20, opacity=0.5]
        (V) -- (W) -- (Wxy) -- (Vxy) -- cycle;
        
        \fill[red!20, opacity=0.5]
        (W) -- (T) -- (Txy) -- (Wxy) -- cycle;
        
        \fill[green!20, opacity=0.5]
        (P) -- (Q) -- (U) -- (T) -- cycle;
        
        \fill[green!20, opacity=0.5]
        (Q) -- (R) -- (V) -- (U) -- cycle;
        
        \fill[green!20, opacity=0.5]
        (R) -- (S) -- (W) -- (V) -- cycle;
        
        \fill[green!20, opacity=0.5]
        (S) -- (P) -- (T) -- (W) -- cycle;
        
        \end{tikzpicture}
        \centering
        \par\smallskip
        {\small (a)}
        \end{minipage}
    }
    \hfill
    \makebox[0.48\columnwidth][c]{
        \begin{minipage}{\linewidth}
        \tdplotsetmaincoords{65}{135}
        \begin{tikzpicture}[tdplot_main_coords, scale=0.35,
            cube/.style={very thick,black},
            grid/.style={very thin,gray},
            axis/.style={->,blue,ultra thick},
            rotated axis/.style={->,blue,ultra thick}]
        
        
        \def\N{40}          
        \def\zmax{7}        
        \def\dx{0.323}       
        \def\dy{0.85}       
        
        \def\Ax{1.5}
        \def\Ay{0.8}
        
        
        \foreach \i in {0,...,\N} {
        
            \pgfmathsetmacro{\zz}{\zmax*\i/\N}
        
            \pgfmathsetmacro{\xx}{\Ax*sin(50*\zz)}
            \pgfmathsetmacro{\yy}{\Ay*cos(35*\zz)}
        
            \coordinate (A\i) at ({\xx-\dx},{\yy-\dy},\zz);
            \coordinate (B\i) at ({\xx+\dx},{\yy-\dy},\zz);
            \coordinate (C\i) at ({\xx+\dx},{\yy+\dy},\zz);
            \coordinate (D\i) at ({\xx-\dx},{\yy+\dy},\zz);
        }

        
        \foreach \i [evaluate=\i as \j using int(\i+1)] in {0,...,\numexpr\N-1\relax} {
        
            \fill[green!20, opacity=0.55]
                (A\i)--(B\i)--(B\j)--(A\j)--cycle;
        
            \fill[green!20, opacity=0.55]
                (B\i)--(C\i)--(C\j)--(B\j)--cycle;
        
            \fill[green!20, opacity=0.55]
                (C\i)--(D\i)--(D\j)--(C\j)--cycle;
        
            \fill[green!20, opacity=0.55]
                (D\i)--(A\i)--(A\j)--(D\j)--cycle;
        }

        
        \draw[dashed, thin]
            (A0)
            \foreach \i in {1,...,\N} { -- (A\i) };
        
        \draw[dashed, thin]
            (B0)
            \foreach \i in {1,...,\N} { -- (B\i) };
        
        \draw[dashed, thin]
            (C0)
            \foreach \i in {1,...,\N} { -- (C\i) };
        
        \draw[dashed, thin]
            (D0)
            \foreach \i in {1,...,\N} { -- (D\i) };

        
        \fill[green!20, opacity=0.55]
            (A0)--(B0)--(C0)--(D0)--cycle;
        
        \fill[green!20, opacity=0.55]
            (A\N)--(B\N)--(C\N)--(D\N)--cycle;
        
        \draw[thick]
            (A0)--(B0)--(C0)--(D0)--cycle;
        
        \draw[thick]
            (A\N)--(B\N)--(C\N)--(D\N)--cycle;

        \def\axisoffset{1.0}
        
        \draw[rotated axis] (-\axisoffset,0,0) -- (3-\axisoffset,0,0) node[anchor=north east]{$z$};
        \draw[rotated axis] (-\axisoffset,0,0) -- (-\axisoffset,3,0) node[anchor=north west]{$y$};
        \draw[rotated axis] (-\axisoffset,0,0) -- (-\axisoffset,0,8) node[anchor=south]{$x$};
        
        \end{tikzpicture}
        \centering
        \par\smallskip
        {\small (b)}
        \end{minipage}
    }
    \caption{(a) and (b) illustrate bodies satisfying the equality conditions in Theorem~\ref{StrongVolLW_Equality}.}
    \label{fig: VolumeStrongLW_equality}
\end{figure}

\section{An Extremal Graph Theory Problem via Fractional Shearer's Lemma}
We present an application of fractional Shearer lemma, i.e., the weak form in \cite{MadimanT10}, to an extremal graph-theoretic problem. Let $n\in\mathbb{N}$, let $\mathcal{F}$ be a prescribed family of nonempty subsets of $[1:n]$, and let $G$ be a simple undirected graph with vertex set $V(G)\subseteq\{-1,1\}^n$. For distinct vertices $x^n,y^n\in V(G)$, the vertices $x^n$ and $y^n$ are adjacent if and only if the set of coordinates on which they differ belongs to $\mathcal{F}$; that is,
\begin{align}\label{edgedefn}
    \{x^n,y^n\}\in E(G)
    \quad\Leftrightarrow\quad
    \Delta(x^n,y^n)\in\mathcal{F},
\end{align}
where
$\Delta(x^n,y^n)=\{i\in[1:n]:x_i\neq y_i\}$.

We are interested in estimating how large $|E(G)|$ can be as a function of $|V(G)|$ and relevant combinatorial properties of $\mathcal{F}$. For each $S\in\mathcal{F}$, let $E_S(G)$ denote the set of edges $\{x^n,y^n\}$ satisfying $\Delta(x^n,y^n)=S$. We first establish bounds on weighted combinations of the pattern-specific edge counts $|E_S(G)|$. These bounds, in turn, yield an upper bound on the total number of edges $|E(G)|=\sum_{S\in\mathcal{F}}|E_S(G)|.$

Because the Hamming distance between $x^n$ and $y^n$ equals $|\Delta(x^n,y^n)|$, our model recovers earlier distance-based models~\cite{Sason22,BoucheronLM13}. In particular, the choice $\mathcal{F}=\{S\subseteq[1:n]:1\leq |S|\leq\tau\}$ gives the model studied by Sason~\cite{Sason22}, while its specialization to $\tau=1$ gives the setting considered by Boucheron \emph{et al.}~\cite[Chapter~4]{BoucheronLM13}. 

Our model is also inspired by a conceptually related transition from cardinality-based constraints to family-based constraints in secure multiparty computation (MPC). The BGW protocol~\cite{BenOrGW88} is formulated under a threshold adversary model, in which corruptible coalitions are specified by a cardinality constraint, whereas Hirt and Maurer~\cite{HirtM97} studied general adversary structures $\mathcal{A}$ that specify which coalitions may be corrupted.

Building on the confusability graph viewpoint for the model in \cite{Sason22}, our setup admits a more general interpretation. Consider a noisy channel with binary $n$-length sequences as inputs. Suppose two input sequences $x^n,y^n\in V(G)$ may produce a common channel output with positive probability if and only if $\Delta(x^n,y^n)\in \mathcal{F}$. In the models considered in \cite{BoucheronLM13} and \cite{Sason22}, confusability depends only on the Hamming distances due to the specific choices of $\mathcal{F}$, whereas in our setting it potentially depends on the full pattern of coordinates in which they differ. This distinction is reflected in our use of the fractional version of Shearer's lemma, whereas Han's inequality suffices for the symmetric, cardinality-based models considered in \cite{BoucheronLM13}, \cite{Sason22}.

We need the following definitions. For any $x^n = (x_1,\ldots,x_n) \in \{-1,1\}^n, S\subseteq [1:n]$, let $\Tilde{x}^{(S)}$ denote an $n-|S|$ length vector obtained by dropping the bits corresponding to positions in $S$ and let $\Bar{x}^{(S)}$ denote an $n$-length vector obtained by flipping the bits corresponding to the specific positions in $S$ in $x^n$. That is, for $S=\{i_1,\ldots,i_d\}$ with $i_1<\ldots<i_d$, 
$\Tilde{x}^{(S)} := (x_1,\ldots,x_{i_1-1},x_{i_1+1},\ldots,x_{i_d-1},x_{i_d+1},\ldots,x_n),$ and $\Bar{x}^{(S)}$ is given by $\Bar{x}^{(S)}_j = -x_j$ for $j \in S$, and $\Bar{x}^{(S)}_j = x_j$ otherwise. Further, for $S \subseteq [1:n]$, we define two integers $m_S$ and $\ell_S$ as follows.
\begin{align}\label{eq: m_S}
    2 \leq m_S \leq \min_{\substack{x^n \in V(G):\\ \bar{x}^{(S)} \in V(G)}} \left|\left\{ y^n \in V(G) : \Tilde{y}^{(S)} = \Tilde{x}^{(S)} \right\}\right|,    
\end{align}
\begin{align}\label{eq: l_S}
    1 \leq \ell_S \leq \min_{\substack{x^n \in V(G):\\ \bar{x}^{(S)} \notin V(G)}} \left|\left\{ y^n \in V(G) : \Tilde{y}^{(S)} = \Tilde{x}^{(S)} \right\}\right|. 
\end{align}

\begin{theorem}\label{ExtremalGraph}
    Let $\mathcal{F}$ be a family of subsets of $[1:n]$, and let $G$ be a simple undirected graph with set of vertices $V(G) \subseteq \{-1,1\}^n, n \in \mathbb{N}$ and set of edges $E(G)=\{\{x^n,y^n\}:x^n,y^n\in V(G) \ \text{s.t.}\ \Delta(x^n,y^n)\in \mathcal{F}\}$. Let integers $m_S$ and $\ell_S$ be as defined in \eqref{eq: m_S} and \eqref{eq: l_S} respectively, then we have the following.
    \begin{enumerate}[leftmargin=*]
        \item Let $\alpha:\bar{\mathcal{F}}\rightarrow \mathbb{R}_+$ be a fractional covering with respect to the complementary family $\bar{\mathcal{F}}=\{S^{\mathrm{c}}: S\in\mathcal{F}\}$. Then 
        \begin{align}
            \sum_{S \in \mathcal{F}} \frac{\alpha(S^{\mathrm{c}})}{w(\alpha)-1} \left( \frac{2|E_S(G)|}{|V(G)|} \log{\frac{m_S}{\ell_S}} + \log{\ell_S}\right) \nonumber \\
            \leq \log{|V(G)|}, \label{eq: EGBound}
        \end{align}
        where, for each $S\in\mathcal{F}$, $E_S(G)$ denotes the set of edges $\{x^n,y^n\}$ such that $\Delta(x^n,y^n)=S$.

        \item Let $2\leq m_{\mathcal{F}}\leq \min_{S\in\mathcal{F}}m_S$ and $1\leq\ell_{\mathcal{F}}\leq \min_{S\in\mathcal{F}}\ell_S$. If $\log\frac{m_{\mathcal{F}}}{\ell_{\mathcal{F}}}>0$
        \begin{align}
            &\hspace{-0.5em}|E(G)|=\sum_{S\in\mathcal{F}}|E_S(G)| \nonumber \\
            & \hspace{-0.55em}\leq \frac{|V(G)|}{2}\min \left\{\frac{ r(\mathcal{F})\log{|V(G)|} - |\mathcal{F}| \log{\ell_{\mathcal{F}}}}{ \log\frac{m_{\mathcal{F}}}{\ell_{\mathcal{F}}}}, |\mathcal{F}|\right\},
        \end{align}
        where $r(\mathcal{F}) = \max_{i \in [1:n]} |\{S \in \mathcal{F}: i \in S\}|$.
        \item Setting the minimum admissible values of $m_{\mathcal{F}} = 2$ and $\ell_{\mathcal{F}} = 1$ yields 
        \begin{align}\label{eq: EGBound_Default}
            |E(G)| & \leq  \frac{|V(G)|}{2} \min\left\{r(\mathcal{F})\log{|V(G)|},|\mathcal{F}|\right\}.
        \end{align}
    \end{enumerate}
\end{theorem}

\begin{remark}\label{BLMSasonRemark}
Theorem~\ref{ExtremalGraph} recovers the bound in \cite[Theorem~4.2]{BoucheronLM13} by taking $\mathcal{F}=\{\{i\}:i\in[1:n]\}$. It also recovers the bounds in \cite[Theorem~2]{Sason22} by first applying the theorem with $\mathcal{F}=\{S\subseteq[1:n]:|S|=d\}$ separately for each
$d\in[1:\tau]$ and then summing the resulting bounds.
\end{remark}

\begin{example}
    Let $1\leq k<n$, and let $\mathcal{F}$ be the family of all cyclically
    consecutive subsets of $[1:n]$ of cardinality $k$, i.e., let
    $S_i=\{i,i+1,\ldots,i+k-1\}$, with indices interpreted modulo $n$, and
    set $\mathcal{F}=\{S_i:i\in[1:n]\}$. Consider the graph $G$ defined by
    \eqref{edgedefn}, and suppose that $|V(G)|<2^{n/k}$.

    We have $|\mathcal{F}|=n$, and every index belongs to exactly $k$
    members of $\mathcal{F}$, so $r(\mathcal{F})=k$. Hence, Theorem~\ref{ExtremalGraph} gives
    \begin{align}
        |E(G)|
        \leq \frac{k|V(G)|\log|V(G)|}{2}.
    \end{align}
\end{example}
\if \extended 1%
A detailed proof of Theorem~\ref{ExtremalGraph} is provided in Appendix~\ref{appendix:f}.
\fi

\subsection{An Edge-Isoperimetric Inequality on the Binary Hypercube}
We next apply Theorem~\ref{ExtremalGraph} to derive an edge-isoperimetric inequality concerning the \emph{influences} of subsets of the binary hypercube~\cite[Chapter~4]{BoucheronLM13}. Given a subset $\mathcal{A}\subseteq\{-1,1\}^n$ and a prescribed family $\mathcal{F}$ of subsets of $[1:n]$, we consider, for each $S\in\mathcal{F}$, how simultaneously flipping the coordinates indexed by $S$ affects membership in $\mathcal{A}$, and then average the resulting influences over $\mathcal{F}$.

 Consider a random vector $X^n=(X_1,X_2,\dots,X_n)$ distributed uniformly over $\{-1,1\}^n$. Let $\bar{X}^{(S)}$ denote the vector obtained by flipping the bits corresponding to the set of positions $S \subseteq [1:n]$, i.e., $\Bar{X}^{(S)}$ is given by $\Bar{X}^{(S)}_j = -X_j$ for $j \in S$, and $\Bar{X}^{(S)}_j = X_j$ otherwise. We define the influence associated with $S \subseteq [1:n]$ on $\mathcal{A}\subseteq \{-1,1\}^n$ as the probability that toggling the bits corresponding to $S$ changes the membership of $X^n$ in $\mathcal{A}$. Formally,
\begin{align}
    I_S(\mathcal{A}) := \mathbb{P}(\mathbbm{1}\{X^n \in \mathcal{A}\} \neq \mathbbm{1}\{\bar{X}^{(S)} \in \mathcal{A}\}).
\end{align}
For a family $\mathcal{F}$ of subsets of $[1:n]$, we define the average influence of subsets of random variables corresponding to the family $\mathcal{F}$ on $\mathcal{A}$ by
\begin{align}
    I_{\mathcal{F}}(\mathcal{A}) := \frac{1}{|\mathcal{F}|} \sum_{S \in \mathcal{F}} I_S(\mathcal{A}).
\end{align}
The following theorem provides a lower bound on $I_{\mathcal{F}}(\mathcal{A})$.

\begin{theorem}\label{InfluenceLowerBound}
    Let $X^n$ be a uniform random vector over $\{-1,1\}^n$, and $\mathcal{A} \subseteq \{-1,1\}^n$. Let $\mathcal{F}$ be any family of subsets of $[1:n]$ and let $m_{\mathcal{F}}$ and $\ell_{\mathcal{F}}$ be admissible integers as in Theorem~\ref{ExtremalGraph}, with $\mathcal{A}$ in place of $V(G)$. If $\log\frac{m_{\mathcal{F}}}{\ell_{\mathcal{F}}}>0$, then 
    \begin{align}
        I_{\mathcal{F}}(\mathcal{A}) \geq \frac{L|\mathcal{A}|} {2^{n-1}\log\frac{m_{\mathcal{F}}}{\ell_{\mathcal{F}}}},
    \end{align}
    where $L = \max\left\{\log{m_{\mathcal{F}}} - \frac{r(\mathcal{F})}{|\mathcal{F}|} \log{|\mathcal{A}|}, 0\right\}$ and $r(\mathcal{F}) = \max_{i \in [1:n]} |\{S \in \mathcal{F}: i \in S\}|$.
\end{theorem}

\begin{remark}
Theorem~\ref{InfluenceLowerBound} recovers the bound in \cite[Theorem~4.3]{BoucheronLM13} by taking $\mathcal{F}=\{\{i\}:{i}\in[1:n]\}$, and the bound in
\cite[Theorem~3]{Sason22} by taking $\mathcal{F}=\{S\subseteq[1:n]:|S|=d\}$.
\end{remark}

\if \extended 1%
A detailed proof of Theorem~\ref{InfluenceLowerBound} is provided in Appendix~\ref{appendix:g}.
\fi
The special case of Theorem~\ref{InfluenceLowerBound} corresponding to \cite[Theorem~4.3]{BoucheronLM13} is known to improve upon the bound obtained from the Efron--Stein inequality in certain regimes~\cite[Remark~4.2]{BoucheronLM13}, and it also follows from the logarithmic Sobolev inequality~\cite[Section~9.2]{BoucheronLM13}. These connections suggest that the family-based inequality established in Theorem~\ref{InfluenceLowerBound} may fit within a broader functional-inequality framework, the investigation of which constitutes an interesting direction for future work. Such a direction may also be relevant to Boolean functions, game theory, complexity theory, and other areas in which influences on the binary hypercube play a fundamental role~\cite[Chapters~9 and 10]{BoucheronLM13}.






\clearpage

\if \extended 1
\IEEEtriggeratref{13}
\fi
\if \extended 0
\IEEEtriggeratref{13}
\fi

\bibliographystyle{IEEEtran}
\bibliography{bibliofileC}

\clearpage


\if \extended 1

\appendices

\section{Proof of Theorem~\ref{Gap_Decomposition}}\label{appendix:a}
\emph{Proof of Part 1).} 
\begin{align}
    &\operatorname{Gap}(\mathcal{F},\gamma) \nonumber \\
    & =  \sum\limits_{S \in \mathcal{F}} \gamma(S)H(X_S|X_{<S}) - H(X_{[1:n]}) \\
    & = \sum\limits_{S \in \mathcal{F}} \gamma(S) \sum\limits_{i \in S} H(X_i|X_{[1:i-1] \cap S},X_{<S}) \nonumber \\
    & \quad\quad - \sum\limits_{i = 1}^n H(X_i|X_{[1:i-1]}) \label{eq: GD_7}\\
    & = \sum\limits_{S \in \mathcal{F}} \gamma(S) \sum\limits_{i \in S} H(X_i|X_{[1:i-1] \cap S},X_{<S}) \nonumber \\
    & \quad\quad - \sum\limits_{i = 1}^n H(X_i|X_{[1:i-1]}) \left( \sum\limits_{\substack{S \in \mathcal{F}:\\ i \in S}} \gamma(S) \right) \label{eq: GD_8} \\
    & = \sum\limits_{S \in \mathcal{F}} \gamma(S) \sum\limits_{i \in S} H(X_i|X_{[1:i-1] \cap S},X_{<S}) \nonumber \\
    & \quad\quad - \sum\limits_{S \in \mathcal{F}} \gamma(S) \sum\limits_{i \in S}H(X_i|X_{[1:i-1]}) \label{eq: GD_9}\\
    &= \sum_{S\in\mathcal{F}}\gamma(S) \sum_{i\in S} I\left( X_i;X_{[1:i-1] \cap S^{\mathrm{c}} \setminus{(<S)}} | X_{([1:i-1] \cap S) \cup (<S)} \right),
\end{align}
where \eqref{eq: GD_7} follows from the chain rule of entropy, \eqref{eq: GD_8} follows because $\sum\limits_{\substack{S \in \mathcal{F}:\\ i \in S}} \gamma(S) = 1$, and \eqref{eq: GD_9} follows by interchanging the summations.

\emph{Proof of Part 2).}
\begin{align}
    &\tbinom{n-1}{k-1}\operatorname{Gap}(\mathcal{C}_k,\gamma) \nonumber \\
    & = \sum\limits_{\substack{S \subseteq [1:n]:\\|S| = k}} H(X_S|X_{<S}) - \tbinom{n-1}{k-1}H(X_{[1:n]}) \\
    & = \sum\limits_{j = 1}^{n-k+1} \sum\limits_{A \in \mathcal{F}_j} H(X_j, X_A|X_{[1:j-1]}) - \tbinom{n-1}{k-1}H(X_{[1:n]}) \label{eq: GD_1} \\
    & = H(X_{[n-k+1:n]}|X_{[1:n-k]}) \nonumber \\
    &~+ \sum\limits_{j = 1}^{n-k} \sum\limits_{A \in \mathcal{F}_j} H(X_j, X_A|X_{[1:j-1]}) - \tbinom{n-1}{k-1}H(X_{[1:n]}) \\
    & = H(X_{[n-k+1:n]}|X_{[1:n-k]}) \nonumber \\
    &~+ \sum\limits_{j = 1}^{n-k} \sum\limits_{A \in \mathcal{F}_j} \left( H(X_j|X_{[1:j-1]}) + H( X_A|X_{[1:j]}) \right) \nonumber \\
    &~- \tbinom{n-1}{k-1}H(X_{[1:n]}) \label{eq: GD_2} \\
    & = H(X_{[n-k+1:n]}|X_{[1:n-k]}) \nonumber \\
    &~+ \sum\limits_{j = 1}^{n-k} \left( \tbinom{n-j}{k-1} H(X_j|X_{[1:j-1]}) + \tbinom{n-j-1}{k-2} H(X_{[j+1:n]}|X_{[1:j]})\right) \nonumber \\
    &~+ \sum\limits_{j = 1}^{n-k} \left( \sum\limits_{A \in \mathcal{F}_j} H( X_A|X_{[1:j]}) - \tbinom{n-j-1}{k-2} H(X_{[j+1:n]}|X_{[1:j]})\right) \nonumber \\
    &~- \tbinom{n-1}{k-1}H(X_{[1:n]}) \\
    & = \sum\limits_{j = 1}^{n-k} \frac{\binom{n-j-1}{k-2}}{\binom{n-1}{k-1}}(\mathcal{F}_j,\gamma_j)\emph{-MI}(X_{j+1}; \ldots; X_n \mid X_{[1:j]}) \label{eq: GD_3},
\end{align}
where \eqref{eq: GD_1} follows because every $k$-sized subset of $[1:n]$ whose minimum element is $j \in [1:n]$ can be uniquely expressed as $\{j\} \cup A$ for some $A \in \mathcal{F}_j$, \eqref{eq: GD_2} follows from chain rule for entropy, and \eqref{eq: GD_3} follows because
\begin{align}
    &\sum\limits_{j = 1}^{n-k} \left( \tbinom{n-j}{k-1} H(X_j|X_{[1:j-1]}) + \tbinom{n-j-1}{k-2} H(X_{[j+1:n]}|X_{[1:j]})\right) \\
    & = \sum\limits_{j = 1}^{n-k} \tbinom{n-j}{k-1} H(X_j|X_{[1:j-1]}) \nonumber \\
    & \quad + \sum\limits_{j = 1}^{n-k} \left(\tbinom{n-j}{k-1}- \tbinom{n-j-1}{k-1}\right) H(X_{[j+1:n]}|X_{[1:j]}) \label{eq: GD_4}\\
    & = \sum\limits_{j = 1}^{n-k} \tbinom{n-j}{k-1} H(X_{[j:n]}|X_{[1:j-1]}) \nonumber \\
    & \quad - \sum\limits_{j = 1}^{n-k} \tbinom{n-j-1}{k-1} H(X_{[j+1:n]}|X_{[1:j]}) \label{eq: GD_5} \\
    & = \tbinom{n-1}{k-1}H(X_{[1:n]}) - H(X_{[n-k+1:n]}|X_{[1:n-k]}) \label{eq: GD_6},
\end{align}
where \eqref{eq: GD_4} follows from Pascal's rule, \eqref{eq: GD_5} follows by chain rule for entropy, and \eqref{eq: GD_6} follows because the expression in \eqref{eq: GD_5} is a telescoping sum.

\section{Proof of Corollary~\ref{Strong_Han_Gap_Stability}}\label{appendix:b}
Using Theorem~\ref{Gap_Decomposition}, we get that
\begin{align}\label{eq: SHGS_5}
    &\operatorname{Gap}(\mathcal{C}_k,\gamma) \nonumber \\
    &= \sum\limits_{j = 1}^{n-k} \frac{\binom{n-j-1}{k-2}}{\binom{n-1}{k-1}}(\mathcal{F}_j,\gamma_j)\text{-MI}(X_{j+1}; \ldots; X_n \mid X_{[1:j]}),
\end{align}
where $\mathcal{F}_j$, for $j \in [1:n-k]$, denotes the family of all $(k-1)$-sized subsets of $[j+1:n]$ and $\gamma_j$ denotes the fractional partition with respect to $\mathcal{F}_j$ such that $\gamma_j(S) = \frac{1}{\binom{n-j-1}{k-2}},$ for all $S \in \mathcal{F}_j$. Putting together the given condition that $\operatorname{Gap}(\mathcal{C}_k,\gamma) \leq \varepsilon$ and \eqref{eq: SHGS_5} gives
\begin{align}\label{eq: SHGS_6}
    (\mathcal{F}_1,\gamma_1)\text{-MI}(X_2; \ldots; X_n \mid X_1) &\leq \frac{\binom{n-1}{k-1}}{\binom{n-2}{k-2}} \varepsilon\nonumber \\
    & = \frac{n-1}{k-1}\varepsilon,
\end{align}
since each summand in \eqref{eq: SHGS_5} is non-negative by Theorem~\ref{Conditional_FGamma_MI}. This completes the proof of Corollary~\ref{Strong_Han_Gap_Stability}.

The reader may assume correctly that there are upper bounds for each summand in \eqref{eq: SHGS_5}. However, this information is subsumed in \eqref{eq: SHGS_6}. 
This is because for any $j \in [1:n-k-1]$, $\binom{n-j-1}{k-2}(\mathcal{F}_j,\gamma_j)\text{-MI}(X_{j+1}; \ldots; X_n \mid X_{[1:j]}) \geq \binom{n-j-2}{k-2}(\mathcal{F}_{j+1},\gamma_{j+1})\text{-MI}(X_{j+2}; \ldots; X_n \mid X_{[1:j+1]})$, as shown below.
\begin{align}
    & (\mathcal{F}_j,\gamma_j)\text{-MI}(X_{j+1}; \ldots; X_n \mid X_{[1:j]}) \nonumber \\
    & = \sum\limits_{S \in \mathcal{F}_j} \frac{1}{\binom{n-j-1}{k-2}} H(X_S|X_{[1:j]}) - H(X_{[j+1:n]}|X_{[1:j]}) \\
    & = \sum\limits_{\substack{S \in \mathcal{F}_j: \\ j+1 \in S}} \frac{1}{\binom{n-j-1}{k-2}} H(X_{j+1}, X_{S \setminus{\{j+1\}}}|X_{[1:j]}) \nonumber \\
    & \quad + \sum\limits_{\substack{S \in \mathcal{F}_j: \\ j+1 \notin S}} \frac{1}{\binom{n-j-1}{k-2}} H(X_S|X_{[1:j]}) - H(X_{[j+1:n]}|X_{[1:j]}) \\
    & = \sum\limits_{\substack{S \in \mathcal{F}_j: \\ j+1 \in S}} \frac{1}{\binom{n-j-1}{k-2}} H(X_{j+1}|X_{[1:j]}) \nonumber \\
    & \quad + \sum\limits_{\substack{S \in \mathcal{F}_j: \\ j+1 \in S}} \frac{1}{\binom{n-j-1}{k-2}} H(X_{S \setminus{\{j+1\}}}|X_{[1:j+1]}) \nonumber \\
    & \quad + \sum\limits_{\substack{S \in \mathcal{F}_j: \\ j+1 \notin S}} \frac{1}{\binom{n-j-1}{k-2}} H(X_S|X_{[1:j]}) \nonumber \\
    & \quad - H(X_{j+1}|X_{[1:j]}) - H(X_{[j+2:n]}|X_{[1:j+1]}) \label{eq: SHGS_1} \\
    & = \sum\limits_{\substack{S \in \mathcal{F}_j: \\ j+1 \in S}} \frac{1}{\binom{n-j-1}{k-2}} H(X_{S \setminus{\{j+1\}}}|X_{[1:j+1]}) \nonumber \\
    & \quad + \sum\limits_{\substack{S \in \mathcal{F}_j: \\ j+1 \notin S}} \frac{1}{\binom{n-j-1}{k-2}} H(X_S|X_{[1:j]}) \nonumber \\
    & \quad - H(X_{[j+2:n]}|X_{[1:j+1]}) \label{eq: SHGS_2} \\
    & = \sum\limits_{\substack{S \in \mathcal{F}_j: \\ j+1 \in S}} \frac{1}{\binom{n-j-1}{k-2}} H(X_{S \setminus{\{j+1\}}}|X_{[1:j+1]}) \nonumber \\
    & \quad - \left( 1 - \frac{\binom{n-j-2}{k-2}}{\binom{n-j-1}{k-2}} \right) H(X_{[j+2:n]}|X_{[1:j+1]}) \nonumber \\
    & \quad + \sum\limits_{\substack{S \in \mathcal{F}_j: \\ j+1 \notin S}} \frac{1}{\binom{n-j-1}{k-2}} H(X_S|X_{[1:j]}) \nonumber \\
    & \quad - \frac{\binom{n-j-2}{k-2}}{\binom{n-j-1}{k-2}}H(X_{[j+2:n]}|X_{[1:j+1]}) \\
    & = \left( 1 - \frac{\binom{n-j-2}{k-2}}{\binom{n-j-1}{k-2}} \right) (\mathcal{F}_j',\gamma_j')\text{-MI}(X_{j+2}; \ldots; X_n \mid X_{[1:j+1]}) \nonumber \\
    & \quad + \frac{\binom{n-j-2}{k-2}}{\binom{n-j-1}{k-2}} (\mathcal{F}_{j+1},\gamma_{j+1})\text{-MI}(X_{j+2}; \ldots; X_n \mid X_{[1:j+1]}), \label{eq: SHGS_3}
\end{align}
where \eqref{eq: SHGS_1} follows from chain rule of entropy, \eqref{eq: SHGS_2} follows because there are $\binom{n-j-1}{k-2}$ many subsets of $[j+1:n]$ which contain $j+1$, and \eqref{eq: SHGS_3} follows as for the family $\mathcal{F}_j' = \{S \in \mathcal{F}_j: j+1 \in S\}$ and fractional partition $\gamma_j'$ such that $\gamma_j'(S) = \frac{1}{\binom{n-j-1}{k-2}-\binom{n-j-2}{k-2}},$ for all $S \in \mathcal{F}_j'$, $\sum\limits_{\substack{S \in \mathcal{F}_j: \\ j+1 \in S}} \frac{1}{\binom{n-j-1}{k-2}-\binom{n-j-2}{k-2}} H(X_S|X_{[1:j+1]}) - H(X_{[j+2:n]}|X_{[1:j+1]}) = (\mathcal{F}_j',\gamma_j')\text{-MI}(X_{j+2}; \ldots; X_n ~|~ X_{[1:j+1]})$. Since $(\mathcal{F}_j',\gamma_j')\text{-MI}(X_{j+2}; \ldots; X_n \mid X_{[1:j+1]}) \geq 0$ by Theorem~\ref{Conditional_FGamma_MI}, \eqref{eq: SHGS_3} implies that for any $j \in [1:n-k-1]$, 
\begin{align}\label{eq: SHGS_4}
    &(\mathcal{F}_j,\gamma_j)\text{-MI}(X_{j+1}; \ldots; X_n \mid X_{[1:j]}) \nonumber \\
    &\geq \frac{\binom{n-j-2}{k-2}}{\binom{n-j-1}{k-2}} (\mathcal{F}_{j+1},\gamma_{j+1})\text{-MI}(X_{j+2}; \ldots; X_n \mid X_{[1:j+1]}).
\end{align}
Instantiating \eqref{eq: SHGS_4} with $j = 1$ gives
\begin{align}
    &(\mathcal{F}_1,\gamma_1)\text{-MI}(X_2; \ldots; X_n \mid X_1) \nonumber \\
    &\geq \frac{\binom{n-3}{k-2}}{\binom{n-2}{k-2}} (\mathcal{F}_2,\gamma_2)\text{-MI}(X_3; \ldots; X_n \mid X_{[1:2]})
\end{align}
and combining it with \eqref{eq: SHGS_6} further gives
\begin{align}
    (\mathcal{F}_2,\gamma_2)\text{-MI}(X_3; \ldots; X_n \mid X_{[1:2]}) \leq \frac{\binom{n-1}{k-1}}{\binom{n-3}{k-2}} \varepsilon,
\end{align}
which is what we would have derived using \eqref{eq: SHGS_5} also by lower bounding every other term by $0$ except for $(\mathcal{F}_2,\gamma_2)\text{-MI}(X_3; \ldots; X_n \mid X_{[1:2]})$. Thus, using \eqref{eq: SHGS_6} and the monotonicity of \eqref{eq: SHGS_4} successively gives an upper bound on $(\mathcal{F}_j,\gamma_j)\text{-MI}(X_{j+1}; \ldots; X_n \mid X_{[1:j]})$, for any $j \in [2:n-k-1]$.

\section{Proof of Corollary~\ref{Strong_Han_Gap_Equality}}\label{appendix:c}
Let $\operatorname{Gap}(\mathcal{C}_k,\gamma)$ be equal to $0$. Substituting $\varepsilon = 0$ in Corollary~\ref{Strong_Han_Gap_Stability}, we get
\begin{align}
    (\mathcal{F}_1,\gamma_1)\text{-MI}(X_2; \ldots; X_n \mid X_1) = 0,
\end{align}
which further holds if $X_2, \ldots, X_n$ are conditionally independent given $X_1$ by Theorem~\ref{Conditional_FGamma_MI}.

For the other direction, suppose that $X_2, \ldots, X_n$ are conditionally independent given $X_1$. This gives that $\text{TC}(X_2, \ldots; X_n | X_1) = 0$, i.e.
\begin{align}
    0 & = \text{TC}(X_2, \ldots; X_n | X_1) \\
    & = \sum\limits_{i = 2}^n H(X_i|X_1) - H(X_{[2:n]}|X_1)
\end{align}
which further implies that for any $j \in [2:n]$,
\begin{align}
    0 & = \sum\limits_{i = 2}^j H(X_i|X_1) + \sum\limits_{i = j+1}^n H(X_i|X_1) - H(X_{[2:j]}|X_1) \nonumber \\
    & \quad\quad - H(X_{[j+1:n]}|X_{[1:j]}) \label{eq: SHGE_3} \\
    & = \sum\limits_{i = j+1}^n H(X_i|X_1) - H(X_{[j+1:n]}|X_{[1:j]}) \label{eq: SHGE_1} \\
    & \geq \sum\limits_{i = j+1}^n H(X_i|X_{[1:j]}) - H(X_{[j+1:n]}|X_{[1:j]}), \label{eq: SHGE_2}
\end{align}
where \eqref{eq: SHGE_3} follows from chain rule of entropy, \eqref{eq: SHGE_1} holds because $X_2, \ldots, X_j$ being conditionally independent given $X_1$ implies that $\sum\limits_{i = 2}^j H(X_i|X_1) = H(X_{[2:j]}|X_1)$, and \eqref{eq: SHGE_2} follows as conditioning reduces entropy. Since by subadditivity of entropy, $H(X_{[j+1:n]}|X_{[1:j]}) \leq \sum\limits_{i = j+1}^n H(X_i|X_{[1:j]})$, \eqref{eq: SHGE_2} implies that for any $j \in [2:n]$,
\begin{align}
    H(X_{[j+1:n]}|X_{[1:j]}) = \sum\limits_{i = j+1}^n H(X_i|X_{[1:j]})
\end{align}
which holds if and only if the random variables $X_{j+1}, \ldots, X_n$ are conditionally independent given $X_{[1:j]}$. From Theorem~\ref{Conditional_FGamma_MI}, we have that for any $j \in [1:n-k-1]$, $X_{j+1}, \ldots, X_n$ are conditionally independent given $X_{[1:j]}$ if $(\mathcal{F}_j,\gamma_j)\text{-MI}(X_{j+1}; \ldots; X_n \mid X_{[1:j]}) = 0$. Thus, we get that $\operatorname{Gap}(\mathcal{C}_k,\gamma) = 0$ by Theorem~\ref{Gap_Decomposition}.

\section{Proof of Theorem~\ref{StrongLW-ddim_Equality}}\label{appendix:d}
To establish the conditions for equality \eqref{eq:StrongLW_equality}, we first present an information-theoretic proof of \eqref{eq:StrongLW-ddim} using \eqref{eq: MT_UB_Strong}. We choose $(X_1, \ldots, X_d)$ uniformly at random from $S$. We then apply \eqref{eq: MT_UB_Strong} to $\mathcal{F} = \{\{[1:d]\setminus{\{i\}}: i \in [1:d]\}$ with the fractional partition $\gamma(S) = \frac{1}{d-1}$ for all $S \in \mathcal{F}$, to obtain an upper bound on $H(X_1, \ldots, X_d)$ as follows.
\begin{align}
    &H(X_1, \ldots, X_d) \nonumber \\
    & \leq \sum\limits_{S \in \mathcal{F}} \gamma(S)H(X_S|X_{< S}) \label{eq: SLWD_5}\\
    & = \frac{1}{d-1} H(X_{[2:d]}|X_1) + \frac{1}{d-1} \sum\limits_{i = 2}^d H(X_{[1:d]\setminus{\{i\}}})\\
    & \leq \frac{1}{d-1} H(X_{[2:d]}|X_1) + \frac{1}{d-1} \sum\limits_{i = 2}^d \log{n_{-i}} \label{eq: SLWD_1}\\
    & \leq \frac{1}{d-1} \Big(\log{n^{\text{max}}_{-1}} + \sum\limits_{i = 2}^d \log{n_{-i}} \Big), \label{eq: SLWD_3}
\end{align}
where \eqref{eq: SLWD_1} follows from the definition of $n_{-i},~i \in [2:d]$, together with the fact that $H(W) \leq \log{|\mathcal{W}|}$, where $\mathcal{W}$ denotes the alphabet of $W$. \eqref{eq: SLWD_3} follows because
\begin{align}
    &H(X_{[2:d]}|X_1) \nonumber \\
    & = \sum\limits_{x} P_{X_1}(x)H(X_{[2:d]}|X_1 = x) \\
    & = \sum\limits_{x} P_{X_1}(x) \log{|\{(x_2, \ldots, x_d):(x,x_2, \ldots, x_d)\in S\}|} \label{eq: SLWD_4prime}\\
    & \leq \sum\limits_{x} P_{X_1}(x) \log{n^{\text{max}}_{-1}} \label{eq: SLWD_4}\\
    & = \log{n^{\text{max}}_{-1}},
\end{align}
where \eqref{eq: SLWD_4prime} holds since, conditioned on $X_1=x$, $(X_2,\dots,X_d)$ is uniform and its alphabet is given by 
\begin{align}\{(x_2, \ldots, x_d):(x,x_2, \ldots, x_d)\in S\},
\end{align}
 and \eqref{eq: SLWD_4} follows from the definition of $n^{\text{max}}_{-1}$. Combining \eqref{eq: SLWD_3} with the fact that $H(X_1,\ldots,X_d) = \log{n}$, we get
\begin{align}\label{eqn:thm51}
    n^{d-1}\leq \left(\prod_{i=2}^{d}n_{-i}\right)n^*_{-1|x_1}.
\end{align}

Clearly, equality in \eqref{eqn:thm51} holds if and only if equality holds in each of \eqref{eq: SLWD_5}, \eqref{eq: SLWD_1} and \eqref{eq: SLWD_3}. By Corollary~\ref{Strong_Han_Gap_Equality}, the equality in \eqref{eq: SLWD_5} implies that $X_2, \ldots, X_d$ are conditionally independent given $X_1$, i.e., for all $x_{[1:d]}$ such that $ P_{X_{1}}(x_1) > 0$,
\begin{align}\label{eq: SLWD_21}
    P_{X_{[2:d]}|X_1}(x_{[2:d]}|x_1) = \prod\limits_{i = 2}^d P_{X_i|X_1}(x_i|x_1). 
\end{align}
For each $x_1\in\mathcal{X}_1$ and $i\in[2:d]$, define
\begin{align}
    \mathcal{X}_{x_1}^{(i)}:=\text{supp}(X_i|X_1=x_1).
\end{align}
Since the conditional distribution of $X_{[2:d]}$ given $X_1=x_1$ is supported on $S_{x_1}$, the conditional independence in \eqref{eq: SLWD_21} gives 
\begin{align}\label{eq: SLWD_9}
    S_{x_1} =  \prod_{i=2}^{d} \mathcal{X}^{(i)}_{x_1}.
\end{align}
Define 
\begin{align}
\mathcal{X}_1=\{x_1\in\mathbb{R}: S_{x_1} \neq \emptyset\}.
\end{align}

Equality in \eqref{eq: SLWD_1} implies that, for every $i\in[2:d]$,
\begin{align}
    H(X_{[1:d]\setminus \{i\}})=\log{n_{-i}}.
\end{align}
Thus, $X_{[1:d]\setminus \{i\}}$ will be uniformly distributed on its support. For any point $x_{[1:d]\setminus\{i\}}$ in this support whose first coordinate is fixed $x_1$, we have
\begin{align}
    P_{X_{[1:d]} \setminus{\{i\}}}(x_1,x_{{[2:d]} \setminus{\{i\}}}) & = \sum_{x_i} P_{X_{[1:d]}}(x_{[1:d]}) \\
    & = \frac{|\mathcal{X}^{(i)}_{x_1}|}{n}, \label{eq: SLWD_8}
\end{align}
where \eqref{eq: SLWD_8} follows from the fact that the set $\mathcal{X}^{(i)}_{x_1}$ is exactly the set $\{x_i: (x_1,x_2, \ldots, x_d) \in S\}$ (due to \eqref{eq: SLWD_9}) and that $X_{[1:d]}$ is uniform. Since $P_{X_{[1:d]}\setminus\{i\}}$ is uniform, \eqref{eq: SLWD_8} implies that 
\begin{align}\label{eqn:eqcard}
|\mathcal{X}^{(i)}_{x}| = |\mathcal{X}^{(i)}_{x_1'}|,
\end{align}
for every $i\in[2:d]$, and all $x_1,x_1'\in\mathcal{X}_1$.

Moreover, \eqref{eq: SLWD_9} and \eqref{eqn:eqcard} imply that
\begin{align}
    |S_{x_1}| =  \prod_{i=2}^{d} |\mathcal{X}^{(i)}_{x_1}|
\end{align}
is independent of $x_1\in\mathcal{X}_1$. Consequently, every nonempty slice has cardinality $n^{\text{max}}_{-1}$. Since $X_{[2:d]}$ conditioned on $X_1=x_1$ is uniformly distributed on $S_{x_1}$, 
\begin{align}
    H(X_{[2:d]}|X_1)=\sum_{x_1\in\mathcal{X}_1}P_{X_1}(x_1)\log{|S_{x_1}|}=\log{n^{\text{max}}_{-1}}.
\end{align}

Thus, equality in \eqref{eq: SLWD_3} follows from equality in \eqref{eq: SLWD_5} and \eqref{eq: SLWD_1}, and therefore imposes no additional condition.

For the other direction, suppose that $S_{x_1} =  \prod_{i=2}^{d} \mathcal{X}^{(i)}_{x_1},$ for all $x_1 \in \mathcal{X}_{1}$ and for distinct $x_1,x_1' \in \mathcal{X}_{1}$, $|\mathcal{X}^{(i)}_{x_1}| = |\mathcal{X}^{(i)}_{x_1'}|$ for all $i\in[2:d]$. We choose $X_{[1:d]}$ uniformly at random from $S$. Conditioned on $X_1=x_1$, $X_{[2:d]}$ is uniformly distributed on $S_{x_1} =  \prod_{i=2}^{d} \mathcal{X}^{(i)}_{x_1}$. Consequently, $X_2,\dots,X_d$ are conditionally independent given $X_1$. Therefore, by Corollary~\ref{Strong_Han_Gap_Equality}, equality holds in \eqref{eq: SLWD_5}.

For each $i\in[2:d]$, for any $x_{[1:d]\setminus\{i\}}$ in the support of $X_{[1:d]\setminus\{i\}}$ whose first coordinate is fixed $x_1$, \eqref{eq: SLWD_8} holds. The probability in \eqref{eq: SLWD_8} is constant over the support since $|\mathcal{X}_{x_1}^{(i)}|$ is independent of $x_1$. Hence $H(X_{[1:d]\setminus\{i\}})=\log n_{-i}$ for $i\in[2:d]$. So, equality holds in \eqref{eq: SLWD_1}.

Finally, equality in \eqref{eq: SLWD_3} follows from equality in \eqref{eq: SLWD_5} and \eqref{eq: SLWD_1} as mentioned earlier. Therefore, combining the equalities in \eqref{eq: SLWD_5}, \eqref{eq: SLWD_1} and \eqref{eq: SLWD_3}, we obtain 
\begin{align}\label{eq: SLWD_22}
    n^{d-1} = \left(\prod_{i=2}^{d}n_{-i}\right)n^{\text{max}}_{-1}.
\end{align}
This completes the proof of Theorem~\ref{StrongLW-ddim_Equality}.

\section{Proof of Theorem~\ref{StrongVolLW_Equality}}\label{appendix:e}
To establish the conditions for equality in \eqref{eq:StrongVolLW}, we first present an information-theoretic proof of \eqref{eq:StrongVolLW} using differential entropy version of \eqref{eq: MT_UB_Strong}. Let
\begin{align}
    V:&=V_d(K),\\
    m(x_1):&=V_{d-1}(K_{x_1}),\\
    T:&=\{x_1\in\mathbb{R}:m(x_1)>0\}.
\end{align}
For $i\in[2:d]$ and $x_{-i}=x_{[1:d]\setminus \{i\}}\in\mathbb{R}^{d-1}$, define the $i^{\text{th}}$-coordinate fiber 
\begin{align}
    K_{x_{-i}}^{(i)}:=\{t\in\mathbb{R}:(x_1,\dots,x_{i-1},t,x_{i+1},\dots,x_d)\in K\},
\end{align}
and let 
\begin{align}
    L_i(x_{-i}):=V_1(K_{x_{-i}}^{(i)}).
\end{align}
Choose $(X_1, \ldots, X_d)$ uniformly at random from $K$. Thus,
\begin{align}
    f_{x_{[1:d]}}(X_{[1:d]})=\frac{\mathbbm{1}_K(x_{[1:d]})}{V},
\end{align}
where $\mathbbm{1}_A(u)$ equals $1$ if $u\in A$ and $0$ otherwise. This gives
\begin{align}
    f_{X_1}(x_1)=\frac{m(x_1)}{V},
\end{align}
and for every $x_1\in T$, 
\begin{align}\label{eqn:uniformcond}
    f_{X_{[2:d]}|X_1}(x_{[2:d]}|x_1)=\frac{\mathbbm{1}_{K_{x_1}}(x_{[2:d]})}{m(x_1)}.
\end{align}
Therefore, conditioned on $X_1=x_1$, $X_{[2:d]}$ is uniformly distributed on $K_{x_1}$. 

Similarly, with $X_{-i}=X_{[1:d]\setminus \{i\}}$, we have
\begin{align}\label{eqn:body6}
    f_{X_{-i}}(x_{-i})=\frac{L_i(x_{-i})}{V}, i\in[2:d],
\end{align}
for $x_{-i}\in \mathbb{R}^{d-1}$.

We apply the differential entropy version of \eqref{eq: MT_UB_Strong} to $\mathcal{F} = \{[1:d]\setminus{\{i\}}: i \in [1:d]\}$ with the fractional partition $\gamma(S)=\frac{1}{d-1}$ for all $S\in \mathcal{F}$ to obtain an upper bound on the differential entropy $h(X_1, \ldots, X_d)$ as follows.
\begin{align}
    &h(X_1, \ldots, X_d) \nonumber \\
    & \leq \sum\limits_{S \in \mathcal{F}} \gamma(S)h(X_S|X_{< S}) \label{eq: SVLWD_5}\\
    & = \frac{1}{d-1} h(X_{[2:d]}|X_1) + \frac{1}{d-1} \sum\limits_{i = 2}^d h(X_{[1:d]\setminus{\{i\}}})\\
    & \leq \frac{1}{d-1} h(X_{[2:d]}|X_1) + \frac{1}{d-1} \sum\limits_{i = 2}^d \log{V_{d-1}(K_{-i})} \label{eq: SVLWD_1}\\
      & \leq \frac{1}{d-1} \Big(\log{M_{-1}(K)} + \sum\limits_{i = 2}^d \log{V_{d-1}(K_{-i})} \Big) \label{eq: SVLWD_3}.
\end{align}
Here, \eqref{eq: SVLWD_1} follows from 
\begin{align}
    h(X_{-i})\leq \log{V_{d-1}(K_{-i})},\ i\in[2:d].
\end{align}
\eqref{eq: SVLWD_3} follows because of the following reason. By the definition of the essential supremum, there exists a Lebesgue null set $N\subseteq \mathbb{R}$ such that 
\begin{align}\label{eqn:esssup1}
    m(x_1)\leq M_{-1}(K),\  \text{for every}\ x_1\in\mathbb{R}\setminus N.
\end{align}
Moreover, since $f_{X_1}(x_1)=\frac{m(x_1)}{V}$, we have
\begin{align}
    \mathbb{P}(X_1\in N)=\frac{1}{V}\int_N m(x_1)\,dx_1=0.
\end{align}
Consequently,
\begin{align}
    h(X_{[2:d]}|X_1) & = \int\limits_{x \in \mathbb{R}} f_{X_1}(x_1) h(X_{[2:d]}|X_1 = x_1)\ dx_1\\
    &=\int_{T\setminus N}\frac{m(x_1)}{V}\log m(x_1)\,dx_1\label{eqn:body1}\\
    &\leq \log{M_{-1}(K)}\int_{T\setminus N}\frac{m(x_1)}{V}\, dx_1\label{eqn:body2}\\
    &=\log{M_{-1}(K)},
\end{align}
where \eqref{eqn:body1} follows from \eqref{eqn:uniformcond} and \eqref{eqn:body2} follows from \eqref{eqn:esssup1}. This proves \eqref{eq: SVLWD_3}. 

Combining \eqref{eq: SVLWD_3} with the fact that $h(X_{[1:d]})=\log{V}$, 
we get
\begin{align}\label{eqn:body3}
    V_d(K)^{d-1}\leq \left(\prod_{i=2}^{d}V_{d-1}(K_{-i})\right) M_{-1}(K).
\end{align}
Clearly,  equality in \eqref{eqn:body3} holds if and only if equality holds in each of \eqref{eq: SVLWD_5}, \eqref{eq: SVLWD_1}, and \eqref{eq: SVLWD_3}.  

By Corollary~\ref{Strong_Han_Gap_Equality}, the equality in \eqref{eq: SVLWD_5} implies that $X_2, \ldots, X_d$ are conditionally independent given $X_1$. Thus, for Lebesgue-almost every $x_1\in T$,
\begin{align}\label{eqn:body4}
    f_{X_{[2:d]}|X_1}(x_{[2:d]}|x_1)=\prod_{i=2}^df_{X_i|X_1}(x_i|x_1).
\end{align}
for Lebesgue-almost every $x_{-i}\in\mathbb{R}^{d-1}$. 

For such an $x_1$, define 
\begin{align}\label{eqn:body5}
    \mathcal{A}_{x_1}^{(i)}:=\{x_i\in\mathbb{R}:f_{X_i|X_1}(x_i|x_1)>0\},\ i\in[2:d].
\end{align}
Now, from \eqref{eqn:uniformcond} and \eqref{eqn:body5}, we have
\begin{align}
    f_{X_{[2:d]}|X_1}(x_{[2:d]}|x_1)>0 \Leftrightarrow x_{[2:d]}\in K_{x_1}
\end{align}
and 
\begin{align}
    \prod_{i=2}^df_{X_i|X_1}(x_i|x_1)>0 \Leftrightarrow x_{2:d}\in\prod_{i=2}^d \mathcal{A}_{x_1}^{(i)},
\end{align}
respectively.

So, \eqref{eqn:body4} implies that 
\begin{align}\label{eqn:body7}
V_{d-1}\left( K_{x_1} \triangle \prod_{i = 2}^d \mathcal{A}^{(i)}_{x_1} \right) = 0.
\end{align}

Equality in \eqref{eq: SVLWD_1} implies that 
\begin{align}
    h(X_{-i})=V_{d-1}(K_{-i}), \ i\in[2:d].
\end{align}
Consequently, $X_{-i}$ is uniformly distributed on $K_{-i}$ up to a null set. Therefore,
\begin{align}
    f_{X_{-i}}(x_{-i})=\frac{\mathbbm{1}_{K_{-i}}(x_{-i})}{V_{d-1}(K_{-i})}
\end{align}
for Lebesgue-almost every $x_{-i}\in\mathbb{R}^{d-1}$.

Comparing this with \eqref{eqn:body6} gives
\begin{align}\label{eqn:body81}
    L_{i}(x_{-i})=\frac{V}{V_{d-1}(K_{-i})}:=\ell_i
\end{align}
for Lebesgue-almost every $x_{-i}\in K_{-i}$.

Fix $i\in[2:d]$. For Lebesgue-almost every $x_1\in T$, Fubini's theorem applied to \eqref{eqn:body6} gives
\begin{align}
    L_i(x_1,x_{[2:d]\setminus\{i\}})=V_1(\mathcal{A}_{x_1}^{(i)})
\end{align}
for Lebesgue-almost every 
\begin{align}
    x_{[2:d]\setminus\{i\}}\in\prod_{\substack{j=2\\j\neq i}}^{d}\mathcal{A}_{x_1}^{(j)}.
\end{align}
Comparing this with \eqref{eqn:body81} gives
\begin{align}
    V_1(\mathcal{A}_{x_1}^{(i)})=\ell_i, \ i\in[2:d].
\end{align}
for Lebesgue-almost every $x_1\in T$.

We will show in the other direction of the proof that equality in \eqref{eq: SVLWD_3} follows from equality in \eqref{eq: SVLWD_5} and \eqref{eq: SVLWD_1}.

For the other direction, we first state a consequence of the assumption that $K$ is a body, which will be proved later. 
\begin{claim}\label{claim}
    Suppose there exist constants $\ell_2,\dots,\ell_d>0$ such that, for Lebesgue-almost every $x_1\in\mathbb{R}$ with $V_{d-1}(K_{x_1})>0$, there exist measurable sets  $\mathcal{A}^{(2)}_{x_1},\ldots,\mathcal{A}^{(d)}_{x_1}\subseteq \mathbb{R}$ satisfying
    \begin{align}\label{eqn:body8}
        V_{d-1}\left( K_{x_1} \triangle \prod_{i = 2}^d \mathcal{A}^{(i)}_{x_1} \right) = 0,
    \end{align}
and for each $i\in[2:d]$, 

\begin{align}\label{eqn:body9}
V_1(\mathcal{A}_{x_1}^{(i)})=\ell_i.
\end{align}
Then for every $i\in[2:d]$, 
\begin{align}\label{eqn:claim10}
    V_{d-1}(K_{-i})=\frac{V}{\ell_i}
\end{align}
and 
\begin{align}\label{eqn:claim15}
L_i(x_{-i})=\ell_i
\end{align}
for Lebesgue-almost every $x_{-i}\in K_{-i}$.
\end{claim}
Now suppose that the conditions for equality in the theorem hold, i.e., the hypothesis of Claim~\ref{claim} holds. By \eqref{eqn:body8} and \eqref{eqn:body9}, for Lebesgue-almost every $x_1\in T$, we have
\begin{align}
    m(x_1)&=V_{d-1}(K_{x_1})\\
    &=V_{d-1}\left(\prod_{i=2}^d\mathcal{A}_{x_1}^{(i)}\right)\label{eqn:body10}\\
    &=\prod_{i=2}^dV_1(\mathcal{A}_{x_1})\label{eq:body11}\\
    &=\prod_{i=2}^d \ell_i\label{eqn:body12},
\end{align}
where \eqref{eqn:body10} follows from \eqref{eqn:body8}, \eqref{eq:body11} follows from the product property of the Lebesgue measure, and \eqref{eqn:body12} follows from \eqref{eqn:body9}. Consequently, 
\begin{align}\label{eqn:body123}
    M_{-1}(K)=\prod_{i=2}^d \ell_i.
\end{align}
Moreover, for Lebesgue-almost every $x_1\in T$, with $(X_1,\ldots,X_d)$ uniformly distributed on $K$, 
\begin{align}
    f_{X_{[2:d]}\mid X_1=x_1}(x_{[2:d]})&=
    \frac{\prod_{i=2}^{d}\mathbbm{1}_{\mathcal{A}_{x_1}^{(i)}}(x_i)}{\prod_{i=2}^{d}\ell_i}\\
    &=\prod_{i=2}^{d}\frac{\mathbbm{1}_{\mathcal{A}_{x_1}^{(i)}}(x_i)}{\ell_i}.
\end{align}
for Lebesgue-almost every $x_{[2:d]}\in\mathbb{R}^{d-1}$. Thus, $X_2,\dots,X_d$ are conditionally independent given $X_1$. Therefore, equality holds in \eqref{eq: SVLWD_5}.

By Claim~\ref{claim} and \eqref{eqn:body6}, for $i\in[2:d]$,
\begin{align}
    f_{X_{-i}}(x_{-i})&=\frac{\ell_i}{V}\mathbbm{1}_{K_{-i}}(x_{-i})\\
    &=\frac{\mathbbm{1}_{K_{-i}}(x_{-i})}{V_{d-1}(K_{-i})}
\end{align}
for Lebesgue-almost every $x_{-i}\in\mathbb{R}^{d-1}$. Hence, $X_{-i}$ is uniformly distributed on $K_{-i}$, and
\begin{align}
    h(X_{-i})=\log{V_{d-1}(K_{-i})}, \ i\in[2:d].
\end{align}
Thus, equality holds in \eqref{eq: SVLWD_1}.

Finally, \eqref{eqn:body123} gives
\begin{align}
    h(X_{[2:d]}|X_1) & = \int_{x \in T} f_{X_1}(x_1) h(X_{[2:d]}|X_1 = x_1)\ dx_1\\
    &=\int_{T}\frac{m(x_1)}{V}\log m(x_1)\,dx_1\\
    &= \log{M_{-1}(K)}\int_{T}\frac{m(x_1)}{V}\, dx_1\\
    &=\log{M_{-1}(K)}.
\end{align}
Thus, equality also holds in \eqref{eq: SVLWD_3}. 

Therefore, combining the equalities in \eqref{eq: SVLWD_5}, \eqref{eq: SVLWD_1}, and \eqref{eq: SVLWD_3}, we obtain  
\begin{align}
        V_d(K)^{d-1} = \left(\prod_{i=2}^{d}V_{d-1}(K_{-i})\right)M_{-1}(K).
\end{align}
It remains to prove Claim~\ref{claim}.
\begin{proof}[Proof of Claim~\ref{claim}]
Fix $i\in[2:d]$. Let $G\subseteq T$ be a full-measure set, i.e., $V_1(T\setminus G)=0$, on which \eqref{eqn:body8} and \eqref{eqn:body9} hold. We first show that 
\begin{align}
L_i(x_{-i})\geq \ell_i,\  x_{-i}\in K_{-i}\label{eqn:claim1}.
\end{align}
Fix $x_{-i}\in K_{-i}$, and $x_i$ such that $x_{[1:d]}\in K$. Since $K$ is a body, we have $K=\overline{\operatorname{int}K}$ and so there exist points $p_r$, $r\in\mathbb{N}$ such that 
\begin{align}
    p_r\rightarrow x_{[1:d]} \ \text{as}\ r\rightarrow \infty.
\end{align}
For every $r$, chose $\rho_r\in(0,\frac{1}{r})$ such that 
\begin{align}
    B(p_r,\rho_r)\subseteq \operatorname{int}K.
\end{align}
Thus, 
\begin{align}\label{eqn:claim2}
    ||p_r-x_{[1:d]}||_2+\rho_r\rightarrow 0.
\end{align}

Write $p_r=(p_{r,1},p_{r,2},\dots,p_{r,d})$, and let $I_r:=(p_{r,1}-\rho_r,p_{r,1}+\rho_r)$ be the projection of $B(p_r,\rho_r)$ onto the first coordinate. For every $t\in I_r$, the slice 
\begin{align}
    B(p_r,\rho_r)_t:=\left\{x_{[2:d]}:(t,x_{[2:d]})\in B(p_r,\rho_r)\right\}
\end{align}
is a $(d-1)$-dimensional open ball of radius $\sqrt{\rho_r^2-(t-p_{r,1})^2}>0$.  Since $B(p_r,\rho_r)\subseteq K$, we have $B(p_r,\rho_r)_t\subseteq K_t$, and hence $V_{d-1}(K_t)>0$. Thus, $I_r\subseteq T$. Since $V_1(T\setminus G)=0$ and $V_1(I_r)=2\rho_r>0$, we have $I_r\cap G\neq\emptyset$. Choose $x_1^{(r)}\in I_r\cap G$ such that the slice $B(p_r,\rho_r)_t$ is nonempty.

For this $x_1^{(r)}$, the slice $B(p_r,\rho_r)_{x_1^{(r)}}$ is a nonempty open subset of $K_{x_1^{(r)}}$. Let $U_r$ denote its projection onto the coordinates indexed by $[2:d]\setminus\{i\}$. Then $U_r$ is a nonempty open subset of $\mathbb{R}^{d-2}$ and hence has positive $(d-2)$-dimensional Lebesgue measure.

For $z=x_{[2:d]\setminus\{i\}}\in\mathbb{R}^{d-2}$, let $F_r(z)$ denote the corresponding $i^{\mathrm{th}}$-coordinate fiber of $K_{x_1^{(r)}}$, i.e., 
\begin{align}
    &F_r(z):= \nonumber \\
    &\left\{u\in\mathbb R:(x_2,\ldots,x_{i-1},u,x_{i+1},\ldots,x_d)\in K_{x_1^{(r)}}\right\},
\end{align}
and let $P_r(z)$ denote the corresponding fiber of $\prod_{j=2}^{d}\mathcal{A}_{x_1^{(r)}}^{(j)}$, i.e.,
\begin{align}
    &P_r(z):= \nonumber \\
    &\left\{u\in\mathbb R:(x_2,\ldots,x_{i-1},u,x_{i+1},\ldots,x_d)\in\prod_{j=2}^{d}\mathcal A_{x_1^{(r)}}^{(j)}\right\}.
\end{align}
Since $x_1^{(r)}\in G$, \eqref{eqn:body8} and Fubini's theorem give 
\begin{align}
    0&=V_{d-1}\left( K_{x_1^{(r)}} \triangle \prod_{j = 2}^d \mathcal{A}^{(j)}_{x_1^{(r)}} \right)\\
    &=\int_{\mathbb{R}^{d-2}}V_1\left(F_r(z)\mathbin{\triangle}P_r(z)\right)\,dz.
\end{align}
Thus, 
\begin{align}
    V_1\left(F_r(z)\mathbin{\triangle}P_r(z)\right)=0
\end{align}
for almost every $z\in U_r$. For every $z\in U_r$, the $i^{\mathrm{th}}$-coordinate fiber of $B(p_r,\rho_r)_{x_1^{(r)}}$ is a nonempty open interval contained in $F_r(z)$. Hence,
\begin{align}
    V_1(F_r(z))>0,
    \qquad z\in U_r.
\end{align}
On the other hand, for Lebesgue-almost every $z\in U_r$, $V_1\left(F_r(z)\mathbin{\triangle}P_r(z)\right)=0$. For any such $z\in U_r$, $P_r(z)$ cannot be empty, since otherwise it would imply that $V_1(F_r(z))=0$. 
Moreover, a fiber of the Cartesian product satisfies
\begin{align}
    P_r(z)=
    \begin{cases}
        \mathcal{A}_{x_1^{(r)}}^{(i)},&z\in\displaystyle\prod_{\substack{j=2\\j\neq i}}^d\mathcal{A}_{x_1^{(r)}}^{(j)},\\[2mm]
        \emptyset,
        &\text{otherwise}.
    \end{cases}
\end{align}
It follows that
\begin{align}
    P_r(z)=\mathcal{A}_{x_1^{(r)}}^{(i)}
\end{align}
for Lebesgue-almost every $z\in U_r$. Since $U_r$ has positive measure, choose one such $z_r\in U_r$. Then
\begin{align}
   V_1( F_r(z_r))&=V_1(P_r(z_r))\\
   &=V_1(\mathcal{A}_{x_1^{(r)}}^{(i)})=\ell_i.
\end{align}
Let $x_{-i}^{(r)}=(x_1^{(r)},z_r)$ be the resulting choice of all coordinates other than the $i^{\mathrm{th}}$ coordinate. Since
$z_r\in U_r$, there exists $u_r\in\mathbb{R}$ such that the point $q_r$ obtained by inserting $u_r$ as its $i^{\mathrm{th}}$ coordinate satisfies 
\begin{align}
    q_r\in B(p_r,\rho_r)\qquad\text{and}\qquad(q_r)_{-i}=x_{-i}^{(r)}.
\end{align}
Therefore,
\begin{align}
    \left\|x_{-i}^{(r)}-(p_r)_{-i}\right\|_2\leq \|q_r-p_r\|_2<\rho_r.
\end{align}
It follows that
\begin{align}
    \left\|x_{-i}^{(r)}-x_{-i}\right\|_2 &\leq\left\|x_{-i}^{(r)}-(p_r)_{-i}\right\|_2+ \left\|(p_r)_{-i}-x_{-i}\right\|_2\\
    &<\rho_r+\left\|p_r-x_{[1:d]}\right\|_2\longrightarrow0,
\end{align}
where the convergence follows from \eqref{eqn:claim2}. Hence,
\begin{align}
    x_{-i}^{(r)}\longrightarrow x_{-i}.
\end{align}
Moreover, directly from the definitions of $F_r(z_r)$ and
$x_{-i}^{(r)}$,
\begin{align}
 F_r(z_r)=K_{x_{-i}^{(r)}}^{(i)}.
\end{align}
Let   $A:=K_{x_{-i}}^{(i)}$.
Since $x_{-i}\in K_{-i}$ and $K$ is compact, $A$ is a nonempty compact subset of $\mathbb{R}$. We claim that, for every $\eta>0$, all sufficiently large $r$ satisfy
\begin{align}\label{eqn:claim6}
    F_r(z_r)\subseteq (K_{x_{-i}^{(r)}}^{(i)})_{\eta},
\end{align}
where
\begin{align}
    (A)_{\eta}:=\{u\in\mathbb{R}:\operatorname{dist}(u,A)<\eta\},
\end{align}
with $\operatorname{dist}(u,A)=\inf_{a\in A}|u-a|$.
Suppose otherwise. Then there exist $\eta>0$, a subsequence $(r_k)$, and points $t_k\in F_{r_k}(z_{r_k})$ such that
\begin{align}\label{eqn:claim4}
    \operatorname{dist}(t_k,A)\geq\eta.
\end{align}
Since
\begin{align}
    F_{r_k}(z_{r_k})=K_{x_{-i}^{(r_k)}}^{(i)},
\end{align}
the point $y_k$ obtained by inserting $t_k$ as the $i^{\mathrm{th}}$ coordinate of $x_{-i}^{(r_k)}$ belongs to $K$. By the compactness of $K$, a further subsequence of $(y_k)$ converges to some $y\in K$. Since $x_{-i}^{(r_k)}\rightarrow x_{-i}$, the limit $y$ is obtained by inserting
some $t\in\mathbb{R}$ as the $i^{\mathrm{th}}$ coordinate of $x_{-i}$. Consequently,
\begin{align}
    t\in K_{x_{-i}}^{(i)}.
\end{align}
Moreover, along the same subsequence, $t_k\rightarrow t$. Hence,
\begin{align}
    \operatorname{dist}(t_k,A)
    \leq |t_k-t|
    \longrightarrow0,
\end{align}
contradicting \eqref{eqn:claim4}. This proves \eqref{eqn:claim6}.

Therefore, for every $\eta>0$, choosing $r$ sufficiently large gives
\begin{align}
    V_1((A)_{\eta})
    \geq V_1(F_r(z_r))
    =\ell_i.
\end{align}
Since $A$ is compact, the sets $A_{1/n}$ decrease to $A$, and $V_1(A_1)<\infty$. By continuity of Lebesgue measure from above,
\begin{align}
    V_1(A)&=\lim_{n\to\infty}V_1(A_{1/n})\\
    &\geq\ell_i.
\end{align}
Thus,
\begin{align}
    L_i(x_{-i})=V_1\left(K_{x_{-i}}^{(i)}\right)\geq\ell_i,
\end{align}
which proves \eqref{eqn:claim1}.

We next compute the measure of the set of points $x_{-i}$ for which the corresponding $i^{\text{th}}$-coordinate fiber of $K$ has positive length.
Define 
\begin{align}
    E_i:=\{x_{-i}\in\mathbb{R}^{d-1}:L_i(x_{-i})>0\}, 
\end{align}
and, for $x_1\in\mathbb{R}$, define 
\begin{align}
    E_{i,x_1}:=\{z\in\mathbb{R}^{d-2}:L_i(x_1,z)>0\}.
\end{align}
For every $x_1\in G$, Fubini's theorem applied to
\eqref{eqn:body8} gives
\begin{align}
    L_i(x_1,z)=\ell_i\mathbbm{1}_{\prod_{\substack{j=2\\j\neq i}}^{d}\mathcal{A}_{x_1}^{(j)}}(z)
\end{align}
for Lebesgue-almost every $z\in\mathbb{R}^{d-2}$.
If $x_1\notin T$, then
\begin{align}
    0 =V_{d-1}(K_{x_1})=\int_{\mathbb{R}^{d-2}}L_i(x_1,z)\,dz.
\end{align}
Therefore, $L_i(x_1,z)=0$ for almost every $z\in\mathbb{R}^{d-2}$, and hence
\begin{align}
    V_{d-2}(E_{i,x_1})=0.
\end{align}
Since $V_1(T\setminus G)=0$, Fubini's theorem now gives
\begin{align}
    V_{d-1}(E_i)&=\int_{\mathbb{R}}V_{d-2}(E_{i,x_1})\,dx_1\\
    &=V_1(T)\prod_{\substack{j=2\\j\neq i}}^{d}\ell_j.\label{eqn:claim7}
\end{align}
Similarly, \eqref{eqn:body8} and \eqref{eqn:body9} give 
\begin{align}
    V&=\int_{\mathbb{R}}V_{d-1}(K_{x_1})\,dx_1\\
    &= V_1(T)\prod_{j=2}^{d}\ell_j.\label{eqn:claim8}
\end{align}
Combining \eqref{eqn:claim7} and \eqref{eqn:claim8}, we obtain
\begin{align}\label{eqn:claim9}
    V_{d-1}(E_i)=\frac{V}{\ell_i}.
\end{align}
By definition, $E_i\subseteq K_{-i}$. Conversely, if $x_{-i}\in K_{-i}$, then \eqref{eqn:claim1} gives
\begin{align}
    L_i(x_{-i})\geq \ell_i>0.
\end{align}
So, $x_{-i}\in E_i$. Therefore,
\begin{align}
    E_i=K_{-i}.
\end{align}
Together with \eqref{eqn:claim9}, this gives
\begin{align}\label{eqn:claim13}
    V_{d-1}(K_{-i})=\frac{V}{\ell_i},
\end{align}
which is \eqref{eqn:claim10}. Finally, by Fubini's theorem,
\begin{align}
    V&=\int_{\mathbb{R}^{d-1}}\int_{\mathbb{R}}\mathbbm{1}_K(x_{-i},u)\,du\,dx_{-i}\\
    &=\int_{K_{-i}}L_i(x_{-i})\,dx_{-i}.
\end{align}
Using \eqref{eqn:claim1} and \eqref{eqn:claim13}, we obtain
\begin{align}
    0&=V-\ell_iV_{d-1}(K_{-i})\\
    &=\int_{K_{-i}}\bigl(L_i(x_{-i})-\ell_i\bigr)\,dx_{-i}.
\end{align}
Since $L_i(x_{-i})-\ell_i\geq0$ on $K_{-i}$, it follows that
\begin{align}
    L_i(x_{-i})=\ell_i
\end{align}
for Lebesgue-almost every $x_{-i}\in K_{-i}$. This proves \eqref{eqn:claim15} and completes the proof of Claim~\ref{claim}.
\end{proof}

\section{Proof of Theorem~\ref{ExtremalGraph}}\label{appendix:f}
\emph{Proof of Part 1).} Let $X^n = (X_1, \ldots, X_n)$ be chosen uniformly at random from $V(G)$. So, $P_{X^n}(x^n) = 1/|V(G)|$, for $x^n \in V(G)$ and $H(X^n) = \log{|V(G)|}$. For any $x^n \in V(G)$ and $S \subseteq [1:n]$, we have either $\Bar{x}^{(S)} \in V(G)$ or $\Bar{x}^{(S)} \notin V(G)$. If $\Bar{x}^{(S)} \in V(G)$, then
\begin{align}
    P_{X_S|\Tilde{X}^{(S)}}&(x_S|\Tilde{x}^{(S)}) \nonumber \\
    & = \frac{1}{|\{y^n \in V(G): \Tilde{y}^{(S)} = \Tilde{x}^{(S)}, \Bar{x}^{(S)} \in V(G)\}|} \\
    & \leq \frac{1}{m_S}, \label{eq: EG_1}
\end{align}
where \eqref{eq: EG_1} holds by definition of $m_S$ in \eqref{eq: m_S}. Similarly, if $\Bar{x}^{(S)} \notin V(G)$, then
\begin{align}\label{eq: EG_2}
    P_{X_S|\Tilde{X}^{(S)}}&(x_S|\Tilde{x}^{(S)}) \leq \frac{1}{\ell_S}
\end{align}
by the definition of $l_S$ in \eqref{eq: l_S}.
Now, let us compute $H(X^n) - H(\Tilde{X}^{(S)})$.
\begin{align}
    & H(X^n) - H(\Tilde{X}^{(S)}) \nonumber \\
    & = H(X_S|\Tilde{X}^{(S)}) \label{eq: EG_3}\\
    & = - \sum\limits_{x^n \in \{-1,1\}^n} P_{X^n}(x^n) \log{\left(P_{X_S|\Tilde{X}^{(S)}}(x_S|\Tilde{x}^{(S)})\right)} \\
    & = - \frac{1}{|V(G)|} \sum\limits_{x^n \in V(G)} \log{\left(P_{X_S|\Tilde{X}^{(S)}}(x_S|\Tilde{x}^{(S)})\right)} \\
    & = - \frac{1}{|V(G)|} \sum\limits_{\substack{x^n \in V(G),\\ \Bar{x}^{(S)} \in V(G)}} \log{\left(P_{X_S|\Tilde{X}^{(S)}}(x_S|\Tilde{x}^{(S)})\right)} \nonumber \\
    & \quad\quad - \frac{1}{|V(G)|} \sum\limits_{\substack{x^n \in V(G),\\ \Bar{x}^{(S)} \notin V(G)}} \log{\left(P_{X_S|\Tilde{X}^{(S)}}(x_S|\Tilde{x}^{(S)})\right)} \\
    & \geq \frac{\log{m_S}}{|V(G)|} \sum\limits_{x^n \in \{-1,1\}^n} \mathbbm{1}\{x^n \in V(G),\Bar{x}^{(S)} \in V(G)\} \nonumber \\
    & \quad \quad + \frac{\log{\ell_S}}{|V(G)|} \sum\limits_{x^n \in \{-1,1\}^n} \mathbbm{1}\{x^n \in V(G),\Bar{x}^{(S)} \notin V(G)\}, \label{eq: EG_4}
\end{align}
where \eqref{eq: EG_3} follows by chain rule, and \eqref{eq: EG_4} follows from \eqref{eq: EG_1} and \eqref{eq: EG_2}. 
We can compute both the summations on the RHS of \eqref{eq: EG_4} in terms of $|V(G)|$, and $|E_S(G)|$, where $E_S(G)$ denotes the set of edges connecting pairs of vertices $x^n,y^n\in V(G)$ such that $\Delta(x^n,y^n) =  S$. The first summation is given by
\begin{align}
    \sum\limits_{x^n \in \{-1,1\}^n} \mathbbm{1}\{x^n \in V(G),\Bar{x}^{(S)} \in V(G)\} = 2|E_S(G)|, \label{eq: EG_7}
\end{align}
because every edge $x^n,\Bar{x}^{(S)} \in E_S(G)$ is counted twice in the above summation. For the second summation, notice that
\begin{align}
    & \sum\limits_{x^n \in \{-1,1\}^n} \mathbbm{1}\{x^n \in V(G),\Bar{x}^{(S)} \in V(G)\} \nonumber \\
    & \quad \quad + \sum\limits_{x^n \in \{-1,1\}^n} \mathbbm{1}\{x^n \in V(G),\Bar{x}^{(S)} \notin V(G)\} \\
    & = \sum\limits_{x^n \in \{-1,1\}^n} \mathbbm{1}\{x^n \in V(G)\} \\
    & = |V(G)|. \label{eq: EG_11}
\end{align}
Now, substituting \eqref{eq: EG_7} in \eqref{eq: EG_11} gives the value of the second summation as
\begin{align}
    &\sum\limits_{x^n \in \{-1,1\}^n} \mathbbm{1}\{x^n \in V(G),\Bar{x}^{(S)} \notin V(G)\} \nonumber \\
    & \quad \quad = |V(G)| - 2|E_S(G)|. \label{eq: EG_8}
\end{align}
Substituting the values of the summations from \eqref{eq: EG_7} and \eqref{eq: EG_8} in \eqref{eq: EG_4} gives
\begin{align}
    &H(X^n) - H(\tilde{X}^{(S)}) \nonumber \\
    & \geq 2 \frac{|E_S(G)|}{|V(G)|} \log{m_S} + \frac{\big(|V(G)| - 2|E_S(G)|\big)}{|V(G)|} \log{\ell_S} \\
    & = \frac{2|E_S(G)|}{|V(G)|} \log{\frac{m_S}{\ell_S}} + \log{\ell_S} \label{eq: EG_9}
\end{align}
Multiplying by $\frac{\alpha(S^{\mathrm{c}})}{w(\alpha)-1}$ and subsequently summing over all $S \in \mathcal{F}$ on both the sides of \eqref{eq: EG_9} then gives
\begin{align}
    &\sum\limits_{S \in \mathcal{F}} \frac{\alpha(S^{\mathrm{c}})}{w(\alpha)-1}\left( H(X^n) - H(\Tilde{X}^{(S)}) \right) \nonumber\\
    & \geq \sum\limits_{S \in \mathcal{F}} \frac{\alpha(S^{\mathrm{c}})}{w(\alpha)-1} \left( \frac{2|E_S(G)|}{|V(G)|} \log{\frac{m_S}{\ell_S}} + \log{\ell_S}\right). \label{eq: EG_13}
\end{align}
Also by Theorem~\ref{FractionalSL}, we have that
\begin{align}
    \sum\limits_{S^{\mathrm{c}} \in \bar{\mathcal{F}}} \alpha(S^{\mathrm{c}})H(\tilde{X}^{(S)}) &\geq H(X^n), \label{eq: EG_12}
\end{align}
Subtracting $w(\alpha)H(X^n)$ on both the sides of \eqref{eq: EG_12} gives 
\begin{align}
    \sum\limits_{S^{\mathrm{c}} \in \bar{\mathcal{F}}} \alpha(S^{\mathrm{c}})\left( H(X^n) - H(\Tilde{X}^{(S)}) \right) &\leq (w(\alpha)-1)H(X^n)
\end{align}
as $w(\alpha) = \sum\limits_{S^{\mathrm{c}} \in \bar{\mathcal{F}}} \alpha(S)$, which further implies
\begin{align}
    \sum\limits_{S \in \mathcal{F}} \frac{\alpha(S^{\mathrm{c}})}{w(\alpha)-1}\left( H(X^n) - H(\Tilde{X}^{(S)}) \right) &\leq H(X^n) \\
    & = \log{|V(G)|}. \label{eq: EG_10}
\end{align}
By putting together \eqref{eq: EG_13} and \eqref{eq: EG_10}, we get
\begin{align}
    \sum\limits_{S \in \mathcal{F}} \frac{\alpha(S^{\mathrm{c}})}{w(\alpha)-1} \left( \frac{2|E_S(G)|}{|V(G)|} \log{\frac{m_S}{\ell_S}} + \log{\ell_S}\right) \leq \log{|V(G)|}. \label{eq: EG_6}
\end{align}

\emph{Proof of Part 2).} It follows from \eqref{eq: EG_6} that
\begin{align}\label{eq: EG_17}
    \left( \min\limits_{S \in \mathcal{F}} \frac{\alpha(S^{\mathrm{c}})}{w(\alpha)-1}\right) \sum\limits_{S \in \mathcal{F}} \left( \frac{2|E_S(G)|}{|V(G)|} \log{\frac{m_S}{\ell_S}} + \log{\ell_S}\right) \nonumber \\
    \leq \log{|V(G)|}.
\end{align}
Let $\mathcal{C}$ be the collection of all possible fractional coverings with respect to the family $\bar{\mathcal{F}}$. Since \eqref{eq: EG_17} is valid for any fractional covering $\alpha$ with respect to $\bar{\mathcal{F}}$, it also holds for the fractional covering $\alpha^* \in \mathcal{C}$ which maximizes $\max\limits_{\alpha' \in \mathcal{C}}\min\limits_{S \in \mathcal{F}} \frac{\alpha'(S^{\mathrm{c}})}{w(\alpha')-1}$, i.e.,
\begin{align}\label{eq: EG_18}
    \left( \max\limits_{\alpha' \in \mathcal{C}}\min\limits_{S \in \mathcal{F}} \frac{\alpha'(S^{\mathrm{c}})}{w(\alpha')-1}\right) \sum\limits_{S \in \mathcal{F}} \left( \frac{2|E_S(G)|}{|V(G)|} \log{\frac{m_S}{\ell_S}} + \log{\ell_S}\right) \nonumber \\
    \leq \log{|V(G)|}.
\end{align}
Further, it can be shown that $\max\limits_{\alpha' \in \mathcal{C}}\min\limits_{S \in \mathcal{F}} \frac{\alpha'(S^{\mathrm{c}})}{w(\alpha')-1} = 1/r(\mathcal{F})$. Consider any $\alpha': \bar{\mathcal{F}} \rightarrow \mathbb{R}_+$. We have that for any $i \in [1:n]$
\begin{align}\label{eq: EG_15}
    \sum\limits_{S \in \mathcal{F}: i \in S} \frac{\alpha'(S^{\mathrm{c}})}{w(\alpha')-1} \geq \left( \min\limits_{S \in \mathcal{F}} \frac{\alpha'(S^{\mathrm{c}})}{w(\alpha')-1} \right) |\{S \in \mathcal{F}: i \in S\}|.
\end{align}
Since, $\alpha'$ is a fractional covering with respect to $\bar{\mathcal{F}}$, i.e., $\sum\limits_{S^{\mathrm{c}} \in \bar{\mathcal{F}}: i \in S^{\mathrm{c}}} \alpha'(S^{\mathrm{c}}) \geq 1$, we have that
\begin{align}
    &\sum\limits_{S \in \mathcal{F}: i \in S} \frac{\alpha'(S^{\mathrm{c}})}{w(\alpha')-1} \nonumber \\
    & = \sum\limits_{S^{\mathrm{c}} \in \bar{\mathcal{F}}: i \notin S^{\mathrm{c}}} \frac{\alpha'(S^{\mathrm{c}})}{w(\alpha')-1} \\
    & = \frac{1}{w(\alpha')-1}\left( \sum\limits_{S^{\mathrm{c}} \in \bar{\mathcal{F}}} \alpha'(S^{\mathrm{c}}) - \sum\limits_{S^{\mathrm{c}} \in \bar{\mathcal{F}}: i \in S^{\mathrm{c}}} \alpha'(S^{\mathrm{c}}) \right) \\
    & \leq \frac{w(\alpha')-1}{w(\alpha')-1} = 1. \label{eq: EG_26}
\end{align}
Then \eqref{eq: EG_26} put together with \eqref{eq: EG_15} implies that for any $i \in [1:n]$
\begin{align}
    \left( \min\limits_{S \in \mathcal{F}} \frac{\alpha'(S^{\mathrm{c}})}{w(\alpha')-1} \right) \leq \frac{1}{|\{S \in \mathcal{F}: i \in S\}|},
\end{align}
which further implies that
\begin{align}
    \left( \min\limits_{S \in \mathcal{F}} \frac{\alpha'(S^{\mathrm{c}})}{w(\alpha')-1} \right) & \leq \min\limits_{i \in [1:n]}\frac{1}{|\{S \in \mathcal{F}: i \in S\}|} \\
    & = \frac{1}{\max\limits_{i \in [1:n]} |\{S \in \mathcal{F}: i \in S\}|} \\
    & = \frac{1}{r(\mathcal{F})}. \label{eq: EG_16}
\end{align}
Since \eqref{eq: EG_16} holds for any $\alpha' \in \mathcal{C}$, it follows that
\begin{align}
    \max\limits_{\alpha' \in \mathcal{C}}\min\limits_{S \in \mathcal{F}} \frac{\alpha'(S^{\mathrm{c}})}{w(\alpha')-1} \leq \frac{1}{r(\mathcal{F})}.
\end{align}
As follows, $\alpha^*(S) = \frac{1}{|F| - r(\mathcal{F})}$ for all $S \in \bar{\mathcal{F}}$, is a valid fractional covering, i.e., for any $i \in [1:n]$,
\begin{align}
    \sum\limits_{S^{\mathrm{c}} \in \bar{\mathcal{F}}: i \in S^{\mathrm{c}}} \alpha^*(S) & = \sum\limits_{S \in \mathcal{F}: i \in S} \frac{1}{|F| - r(\mathcal{F})} \\
    & = \frac{|\{S^{\mathrm{c}} \in \bar{\mathcal{F}}: i \in S^{\mathrm{c}}\}|}{|F| - r(\mathcal{F})}\\
    & = \frac{|\mathcal{F}| - |\{S \in \mathcal{F}: i \in S\}|}{|F| - r(\mathcal{F})} \\
    & \geq 1,
\end{align}
and achieves the bound in \eqref{eq: EG_16}.

Finally, substituting $\max\limits_{\alpha' \in \mathcal{C}}\min\limits_{S \in \mathcal{F}} \frac{\alpha'(S^{\mathrm{c}})}{w(\alpha')-1} = \frac{1}{r(\mathcal{F})}$ and rearranging the terms in \eqref{eq: EG_18}, we get
\begin{align}\label{eq: EG_24}
    \sum\limits_{S \in \mathcal{F}} \frac{2|E_S(G)|}{|V(G)|} \log{\frac{m_S}{\ell_S}} +  \sum\limits_{S \in \mathcal{F}} \log{\ell_S} \leq r(\mathcal{F})\log{|V(G)|},
\end{align}
which is equivalent to
\begin{align}\label{eq: EG_27}
    \sum\limits_{S \in \mathcal{F}} \frac{2|E_S(G)|}{|V(G)|} \log{m_S} +  \sum\limits_{S \in \mathcal{F}} \left(1 - \frac{2|E_S(G)|}{|V(G)|} \right)\log{\ell_S} \nonumber \\
    \leq r(\mathcal{F})\log{|V(G)|}
\end{align}
and using $m_{\mathcal{F}}\leq \min_{S\in\mathcal{F}}m_S$ and $\ell_{\mathcal{F}}\leq \min_{S\in\mathcal{F}}\ell_S$ in \eqref{eq: EG_27} gives
\begin{align}\label{eq: EG_28}
    \log{m_{\mathcal{F}}} \sum\limits_{S \in \mathcal{F}} \frac{2|E_S(G)|}{|V(G)|} + \log{\ell_{\mathcal{F}}} \sum\limits_{S \in \mathcal{F}} \left(1 - \frac{2|E_S(G)|}{|V(G)|} \right) \nonumber \\
    \leq r(\mathcal{F})\log{|V(G)|}.
\end{align}
Rearranging the terms in \eqref{eq: EG_28} then gives the following upper bound on $|E(G)|$ as
\begin{align}
    &|E(G)| \nonumber \\
    &= \sum\limits_{S \in \mathcal{F}} |E_S(G)| \leq \frac{|V(G)|\left( r(\mathcal{F})\log{|V(G)|} - |\mathcal{F}| \log{\ell_{\mathcal{F}}} \right)}{2 \log{\frac{m_{\mathcal{F}}}{\ell_{\mathcal{F}}}}}.\label{eq: EG_20}
\end{align}

Note that the bound in \eqref{eq: EG_20} is only valid when $\log{\frac{m_{\mathcal{F}}}{\ell_{\mathcal{F}}}} > 0$ because if $\log\frac{m_{\mathcal{F}}}{\ell_{\mathcal{F}}}$ is negative, division reverses the inequality in \eqref{eq: EG_20} and hence does not yield an upper bound on $|E(G)|$. But irrespective of this condition, the following trivial bound always holds
\begin{align}\label{eq: EG_25}
    |E(G)| \leq \frac{1}{2}|\mathcal{F}||V(G)|
\end{align}
because every vertex $x^n \in V(G)$ can have at most $|\mathcal{F}|$ many neighbours $y^n \in V(G)$ such that $\Delta(x^n,y^n) \in \mathcal{F}$. Thus, putting together \eqref{eq: EG_20} and \eqref{eq: EG_25} gives
\begin{align}
    &\hspace{-0.5em}|E(G)| \nonumber \\
    & \hspace{-0.55em}\leq \frac{|V(G)|}{2}\min \left\{\frac{\left( r(\mathcal{F})\log{|V(G)|} - |\mathcal{F}| \log{\ell_{\mathcal{F}}} \right)}{ \log\frac{m_{\mathcal{F}}}{l_{\mathcal{F}}}}, |\mathcal{F}|\right\}\label{eq: EG_29}
\end{align}

\emph{Proof of Part 3).} Setting the minimum admissible values of $m_S = 2$ and $\ell_S = 1$ for all $S \in \mathcal{F}$ in \eqref{eq: EG_29} yields
\begin{align}
    |E(G)| & \leq  \frac{|V(G)|}{2} \min\left\{r(\mathcal{F}|\log{|V(G)|},|\mathcal{F}|\right\}.
\end{align}

\section{Proof of Corollary~\ref{InfluenceLowerBound}}\label{appendix:g}
Let $V(G) = A$ and $E(G)$ denote the set of edges that connect $x^n,y^n \in V(G)$ such that $\Delta(x^n,y^n) \in \mathcal{F}$. Let the edge boundary of $V(G)$ corresponding to $\mathcal{F}$ be defined as
\begin{align}\label{eq:boundary_defn}
    &\mathcal{B}_{\mathcal{F}}(V(G)) \nonumber \\
    &\quad:= \{(x^n,y^n): x^n \in V(G), y^n \notin V(G), \Delta(x^n, y^n) \in \mathcal{F}\}.
\end{align}
Since every vertex $x^n \in V(G)$ has $|\mathcal{F}|$ neighbours in $\{-1,1\}^n$ such that $\Delta(x^n, y^n) \in \mathcal{F}$ we have that
\begin{align}
    |\mathcal{F}||V(G)| = 2|E(G)| + |\mathcal{B}_{\mathcal{F}}(V(G))|.
\end{align}
Thus, it follows that
\begin{align}
    &|\mathcal{B}_{\mathcal{F}}(V(G))| \nonumber \\
    &= |\mathcal{F}||V(G)| - 2|E(G)| \\
    & \geq |\mathcal{F}||V(G)| \nonumber \\
    & \quad - |V(G)|\min \left\{\frac{\left( r(\mathcal{F})\log{|V(G)|} - |\mathcal{F}| \log{\ell_{\mathcal{F}}} \right)}{ \log\frac{m_{\mathcal{F}}}{l_{\mathcal{F}}}}, |\mathcal{F}|\right\} \label{eq: ILB_1}, \\
    & = |\mathcal{F}||V(G)|\max\left\{ \frac{\left( \log{m_{\mathcal{F}}} - \frac{r(\mathcal{F})}{|\mathcal{F}|}\log{|V(G)|} \right)}{ \log\frac{m_{\mathcal{F}}}{l_{\mathcal{F}}}}, 0\right\} \label{eq: ILB_6},
\end{align}
where \eqref{eq: ILB_1} follows from part 2) of Theorem~\ref{ExtremalGraph}. Consequently, the influence of family $\mathcal{F}$ is lower bounded as follows.
\begin{align}
    &I_{\mathcal{F}}(V(G)) \nonumber \\
    & = \frac{1}{|\mathcal{F}|} \sum\limits_{S \in \mathcal{F}} I_S(V(G)) \\
    & = \frac{1}{|\mathcal{F}|} \sum\limits_{S \in \mathcal{F}} \mathbb{P}(\mathbbm{1}\{X^n \in V(G)\} \neq \mathbbm{1}\{\bar{X}^{(S)} \in V(G)\}) \label{eq: ILB_4}\\
    & = \frac{2}{|\mathcal{F}|} \sum\limits_{S \in \mathcal{F}} \sum\limits_{x^n \in \{-1,1\}^n} \frac{\mathbbm{1}\{x^n \in V(G),\bar{x}^{(S)} \notin V(G)\}}{2^n} \label{eq: ILB_2} \\
    & = \frac{2}{|\mathcal{F}|2^n} |\mathcal{B}_{\mathcal{F}}(V(G))| \label{eq: ILB_5}\\
    & \geq \frac{|V(G)|} {2^{n-1}} \max\left\{ \frac{\left( \log{m_{\mathcal{F}}} - \frac{r(\mathcal{F})}{|\mathcal{F}|}\log{|V(G)|} \right)}{ \log\frac{m_{\mathcal{F}}}{l_{\mathcal{F}}}}, 0\right\} \label{eq: ILB_3} \\
    & = \frac{|\mathcal{A}|} {2^{n-1}} \max\left\{ \frac{\left( \log{m_{\mathcal{F}}} - \frac{r(\mathcal{F})}{|\mathcal{F}|}\log{|\mathcal{A}|} \right)}{ \log\frac{m_{\mathcal{F}}}{l_{\mathcal{F}}}}, 0\right\} \\
    & = \frac{|\mathcal{A}|} {2^{n-1}\log\frac{m_{\mathcal{F}}}{l_{\mathcal{F}}}} \max\left\{ \log{m_{\mathcal{F}}} - \frac{r(\mathcal{F})}{|\mathcal{F}|}\log{|\mathcal{A}|}, 0\right\}
\end{align}
where \eqref{eq: ILB_2} follows because $X^n$ is uniformly distributed over $\{-1,1\}^n$ and each pair $x^n \in V(G),\bar{x}^{(S)} \notin V(G)$ satisfies $\mathbbm{1}\{X^n \in V(G)\} \neq \mathbbm{1}\{\bar{X}^{(S)} \in V(G)\}$ twice in \eqref{eq: ILB_4}, \eqref{eq: ILB_5} follows from the definition of $\mathcal{B}_{\mathcal{F}}(V(G))$ in \eqref{eq:boundary_defn} and \eqref{eq: ILB_3} follows from \eqref{eq: ILB_6}.



\fi
\end{document}